\documentclass{article}

\usepackage[margin=1in]{geometry}
\usepackage{amssymb,amsmath}
\usepackage{amsthm}
\usepackage{graphicx}
\usepackage{subfigure}
\usepackage{bm}
\usepackage{cite}
\usepackage{color}
\usepackage[T1]{fontenc}

\newtheorem{lem}{Lemma}
\newtheorem{thm}{Theorem}
\newtheorem{cor}{Corollary}

\newtheorem{ass}{Assumption}

\providecommand{\norm}[1]{\left\lVert #1 \right\rVert}
\providecommand{\abs}[1]{\left\lvert #1 \right\rvert}
\providecommand{\proxprojection}[1]{\Pi_{\mathcal{W}}\left( #1 \right)}

\newcommand{\wbar}{\overline{w}}
\newcommand{\zbar}{\overline{z}}
\newcommand{\gbar}{\overline{g}}
\newcommand{\ghat}{\widehat{g}}
\newcommand{\defeq}{\stackrel{\text{def}}{=}}
\newcommand{\constraintset}{\mathcal{W}}
\newcommand{\R}{\mathbb{R}}
\newcommand{\E}{\mathbb{E}}
\newcommand{\what}{\widehat{w}}
\newcommand{\order}[1]{\mathcal{O}\left( #1 \right)}

\newcommand{\lhat}{\hat{l}}
\newcommand{\htilde}{\tilde{h}}
\newcommand{\ybar}{\overline{y}}
\newcommand{\AllReduce}{\textsc{AllReduce} }
\newcommand{\DistributedAveraging}{\textsc{DistributedAveraging}}
\newcommand{\Gossip}{\textsc{Gossip}}
\newcommand{\gradient}{\nabla_{w}}
\newcommand{\xhat}{\widehat{x}}
\newcommand{\wbarhat}{\widehat{\overline{w}}}
\renewcommand{\deg}{\operatorname{deg}}

\newcommand{\Exp}[2]{\E\left[ #2 \right]}
\newcommand{\dprod}[2]{\left\langle #1, #2 \right\rangle}

\begin{document}

\title{\bf Efficient Distributed Online Prediction and Stochastic Optimization with Approximate Distributed Averaging}

\author{Konstantinos~I.~Tsianos and Michael~G.~Rabbat\\
 \\
Department of Electrical and Computer Engineering\\
McGill University, Montr\'{e}al, Qu\'{e}bec, Canada\\
Email: konstantinos.tsianos@gmail.com, michael.rabbat@mcgill.ca}

\maketitle

\begin{abstract}
We study distributed methods for online prediction and stochastic optimization. Our approach is iterative: in each round nodes first perform local computations and then communicate in order to aggregate information and synchronize their decision variables. Synchronization is accomplished through the use of a distributed averaging protocol. When an exact distributed averaging protocol is used, it is known that the optimal regret bound of $\order{\sqrt{m}}$ can be achieved using the distributed mini-batch algorithm of Dekel et al.~(2012), where $m$ is the total number of samples processed across the network. We focus on methods using approximate distributed averaging protocols and show that the optimal regret bound can also be achieved in this setting. In particular, we propose a gossip-based optimization method which achieves the optimal regret bound. The amount of communication required depends on the network topology through the second largest eigenvalue of the transition matrix of a random walk on the network. In the setting of stochastic optimization, the proposed gossip-based approach achieves nearly-linear scaling: the optimization error is guaranteed to be no more than $\epsilon$ after $\order{ \frac{1}{n \epsilon^2} }$ rounds, each of which involves $\order{\log n}$ gossip iterations, when nodes communicate over a well-connected graph. This scaling law is also observed in numerical experiments on a cluster.
\end{abstract}

\section{Introduction}

In order to scale up to very large optimization problems, it is necessary to use distributed methods. Such large-scale problems frequently arise in the context of machine learning. For example, when one wishes to fit a model to a very large (static) dataset, the data can be partitioned across multiple machines and the aim is to obtain an accurate solution as fast as possible by leveraging parallel processing. Alternatively, when data samples arrive to the system in a streaming manner, as in the setting of online learning, and the model is progressively updated after each sample is observed, a distributed system may be required when the rate at which samples arrive is faster than a single machine can handle. 

In either the static or online setting, the nodes in the distributed system must communicate in order to coordinate their computation, and it is important to understand the extent to which this communication affects the overall scaling performance of any method. In many algorithms, each node maintains and updates its own copy of the decision variables, and communication is performed periodically to synchronize the variables across the network. Typically this synchronization operation simplifies to distributed averaging: each node initially holds a local vector and the goal is for all nodes to compute (exactly or approximately) the average of these vectors.

This paper proposes and studies a method for distributed online prediction and stochastic optimization using approximate distributed averaging to coordinate and synchronize the values at different nodes. Previous work has shown that methods using exact distributed averaging achieve the optimal regret bounds for distributed online prediction. When an approximate distributed averaging protocol is used, there is some residual synchronization error in each node's decision variables. The magnitude of this error is inversely related to the amount of communication. We present a general framework for online prediction and stochastic optimization using approximate distributed averaging protocols, and we characterize the performance requirements (in terms of accuracy and latency) that a distributed averaging protocol must satisfy in order to achieve optimal rates. Then we focus on gossip protocols for distributed averaging and precisely characterize the number of gossip iterations that must be performed in order to control the error and meet the requirements mentioned above. Gossip protocols have a number of attractive features, including that they are simple to implement, robust to communication failures, and they can naturally be implemented in an asynchronous manner. For this reason, gossip-based distributed optimization methods have recently received considerable attention~\cite{nedicDistributedOptimization,distrStochSubgrOpt,dualAveragingTAC,FastDistributedGradMethods,TsianosCDC2012,Chen2012}.

Before giving a detailed statement of our contributions, we describe the problems of online prediction and stochastic optimization.

\subsection{Online Prediction}
\label{subsec:onlinePrediction}

In online prediction, a system must sequentially make predictions. We consider the setting where the set of possible predictors $\constraintset$ is a compact convex subset of a real Euclidean vector space. After the system makes a prediction $w(t) \in \constraintset$ for $t \ge 1$, a sample $x(t) \in \mathcal{X}$ is revealed, and the system suffers a loss $f(w(t), x(t))$. The random samples $x(t)$ are assumed to be drawn independently and identically from an unknown distribution over $\mathcal{X}$. After observing $x(t)$ and suffering the corresponding loss, the system updates its predictor to $w(t+1)$, receives a new sample $x(t+1)$, suffers the additional loss $f(w(t+1), x(t+1))$, and the process continues.

The goal of an online prediction system is to minimize the loss accumulated over time. Specifically, after a total of $m$ samples have been processed, the regret accrued by a single-processor online prediction system is given by
\begin{equation}
R_1(m) \defeq \sum_{t=1}^m [ f(w(t), x(t)) - f(w^*, x(t)) ] \;, \label{eqn:R_1}
\end{equation}
where $w^* \defeq \arg\min_{w \in \constraintset} \E_x[ f(w, x) ]$ is the predictor that minimizes the expected regret with respect to the unknown distribution over the samples $\{x(t)\}_{t \ge 1}$. When the loss function is convex and smooth (i.e., $f$ has Lipschitz-continuous gradients), it is known that no first-order method can achieve a bound on the expected regret which is better than $\E[R_1(m)] = \order{ \sqrt{m} }$ \cite{NemirovskyYudinBook,CesaBianchiBook,ShwartzOnlineBook}. At the same time, this performance is achieved by a number of methods~\cite{ShwartzOnlineBook}, including Nesterov's dual averaging algorithm~\cite{nesterovDualOpt}.

Applications of online prediction methods are typically those where predictions must be made in real-time when a query arrives to the system. When queries arrive at a rate which is faster than a single processor can handle then a distributed online prediction system must be used. Consider a general system of $n$ nodes. Each node receives a proportion (say $1/n$) of the queries arriving to the system and makes predictions. After making a prediction $w_i(t)$, node $i$ observes a sample $x_i(t)$ and suffers a loss $f(w_i(t), x_i(t))$. The nodes can individually update their predictors based on their local observations (e.g., using $\gradient f(w_i(t), x_i(t))$), and they can also communicate with other nodes to obtain information about their predictions and observations. 

Consider a general framework where, in the $t$-th round, the network collectively makes a batch of $b \ge n$ predictions $\{w_i(t,s) \colon 1\le i \le n, 1 \le s \le b\}$, where $b$ is an integer multiple of $n$. After making the prediction $w_i(t,s)$, node $i$ receives an observation $x_i(t,s)$, and suffers the corresponding loss. Once the entire batch is processed, the nodes perform some communication, and then the next round begins. During the time that the nodes are communicating, additional queries may arrive to the network, and so the nodes continue to make predictions using their most recent predictor. Let $\mu$ denote the number of additional samples that arrive to the system while nodes are communicating ($\mu / n$ samples arrive at each node). For the general framework just described, performance is measured in terms of the regret
\begin{equation}
R_n(m) \defeq \sum_{i=1}^n \sum_{t=1}^{\frac{m}{b + \mu}} \left(\sum_{s=1}^{\frac{b}{n}} [f(w_i(t,s), x_i(t,s)) - f(w^*, x_i(t,s))] + \sum_{s'=1}^{\frac{\mu}{n}} [f(w_i(t,s), x_i(t,s')) - f(w^*, x_i(t,s'))] \right)\;. \label{eqn:R_nlong}
\end{equation}
The terms in the sum over $s$ correspond to the regret suffered while locally processing samples which will be used to update the prediction, and we have precise control over the number of terms in this sum since we can specify $b$. On the other hand, the terms in the sum over $s'$ correspond to the regret accumulated while communication takes place, and the only way to control this portion of the regret is to limit the time spent communicating.

Clearly, we cannot hope to obtain a better regret bound using $n$ processors than would be obtained using a single processor; i.e., $\E[R_n(m)] = \order{ \sqrt{m} }$.\footnote{One might hope to achieve this rate if communication was arbitrarily faster than the rate at which queries arrive to the network and are processed, in which case every processor could always share every sample with every other processor and they behave overall as a single processor observing all $m$ samples.} At the other extreme, consider a na\"{i}ve approach where nodes do not communicate at all. In this case, $\mu = 0$. However, since each node operates in isolation, it individually accrues a regret which is bounded by $R_1(m/n)$, the regret that a single processor would have accrued using only $m/n$ samples. Thus, the no-communication approach has an overall regret bound of $\E[R_n(m)] = \order{n \sqrt{m/n}} = \order{ \sqrt{n m} }$, which is a factor of $\sqrt{n}$ worse than what we hope to achieve~\cite{DekelMiniBatches}.

It has been shown that the constant hidden in the $\order{ \sqrt{m} }$ bounds for serial online prediction algorithms is proportional to the standard deviation of the random gradients $\gradient f(w, x)$ \cite{Lan2012,xiaoDA}. Based on this observation, Dekel et al.~\cite{DekelMiniBatches} propose the distributed mini-batch algorithm. Specifically, the predictor at each node is held fixed at $w_i(t,1)$ for all predictions in the $t$th round. Then, the gradients $\{\gradient f(w_i(t,1), x_i(t,s)) \colon 1\le i \le n, 1 \le s \le b/n\}$ can be averaged to reduce the standard deviation. Averaging the mini-batch of $b$ gradients across all nodes involves communication, and Dekel et al.~\cite{DekelMiniBatches} propose to use an exact distributed averaging protocol. Once the protocol returns, all nodes hold the exact average gradient, 
\[
\frac{1}{b} \sum_{i=1}^n \sum_{s=1}^{\frac{b}{n}} \gradient f(w_i(t,1), x_i(t,s)) \;.
\]
This distributed mini-batch procedure with exact distributed averaging is shown in~\cite{DekelMiniBatches} to achieve a regret bound of
\[
\E[R_n(m)] = \order{ b + \mu + \sqrt{\frac{b + \mu}{b} \cdot m} } \;.
\]
Thus, by choosing $b$ appropriately (i.e., taking $b = \Theta(m^\rho)$ for any $\rho \in (0, 1/2)$), the distributed mini-batch algorithm achieves an expected regret of $\E[R_n(m)] = \order{ \sqrt{m} }$. 

Although exact distributed averaging enables this approach to achieve the optimal regret rate, it also has some drawbacks. In particular, exact distributed averaging generally requires blocking communications or some form of barrier synchronization across the network. All nodes remain blocked until the operation completes, and consequently, the computation proceeds at the pace of the slowest node. This is undesirable in clusters, where infrastructure is prone to fail and where resources are typically shared by many users. Motivated by this observation, the aim of the present manuscript is to develop distributed methods which use approximate distributed averaging protocols that can be much more flexible and can be implemented to run asynchronously.

\subsection{Stochastic Optimization}
\label{subsec:stochasticOptimization}

The problem of stochastic optimization is closely related to that of online prediction. In stochastic optimization we seek to find a value $w \in \constraintset$ that solves the problem
\begin{align*}
\text{minimize} &\quad F(w) \defeq \E_x[f(w,x)] \\
\text{subject to} &\quad w \in \constraintset \;.
\end{align*}
A serial first-order stochastic optimization method sequentially makes stochastic gradient evaluations, observing $\gradient f(w(t),x(t))$, and uses these to update the decision variable $w(t)$. The performance of a single-processor method is measured by the optimality gap $\Delta_1(T) = \E[F(w(T))] - F(w^*)$ after $T$ updates.\footnote{Note that the expectation is taken in the definition of optimality gap because the predictor $w(T)$ is random since it is a function of the random observations $x(t)$, $t=1,\dots,T$.} When $f(w,x)$ is convex in $w$ for all $x \in \mathcal{X}$ and when the $x(t)$ are i.i.d., it is known~\cite{xiaoDA} that any online prediction algorithm that has an expected regret bound of $\E[R_1(m)] = \order{ \sqrt{m} }$ can be used for stochastic optimization and will achieve an optimality gap of $\Delta_1(T) \le \frac{1}{T} \E[R_1(T)] = \order{ \frac{1}{\sqrt{T}} }$.

Stochastic optimization is useful when one may wish to fit a model $w$ to a very large collection of data $\{x_i\}_{i=1}^N$ by minimizing the average loss $\frac{1}{N} \sum_{i=1}^N f(w, x_i)$ over the entire data set. Rather than performing gradient updates using all of the data, which can be very time-consuming when $N$ is large, a stochastic gradient is obtained by drawing one sample $x(t) = x_i$ independently and uniformly from the collection at each iteration, and using $\gradient f(w(t), x(t))$ as a surrogate for $\frac{1}{N} \sum_{i=1}^N \gradient f(w(t), x_i)$. Clearly, when $x(t)$ is drawn independently and uniformly over $\{x_i\}_{i=1}^N$, we have that 
\[
\E_{x(t)}[f(w(t), x(t))] = \frac{1}{N} \sum_{i=1}^N \gradient f(w(t), x_i) \;.
\]

The motivation to use a distributed system for solving stochastic optimization problems is often to obtain a solution faster than a single processor. In contrast to the online prediction setup, the system does not accrue any additional regret (or suffer an increase in the optimality gap) when communicating. On the other hand, it is of interest to study the scaling properties of distributed methods. For an optimal single-processor method to ensure the optimality gap is bounded by $\Delta_1(T) \le \epsilon$ for some $\epsilon > 0$, it must evaluate $T \propto \frac{1}{\epsilon^2}$ stochastic gradients. An ideal distributed method using $n$ nodes would obtain a linear speedup by guaranteeing the same level of accuracy after a time equivalent to what it would take the single processor to evaluate $\frac{1}{n \epsilon^2}$ stochastic gradients. To achieve this runtime scaling, a distributed method must not require excessive communication per round, relative to the time spent computing.

\subsection{Contributions}

This paper proposes and analyzes a method for online prediction and stochastic optimization. The proposed method builds on the dual averaging algorithm and the idea of using distributed mini-batch computations. Specifically, we consider a framework using general approximate distributed averaging protocols that ensure the dual variables at each node are synchronized to within accuracy $\delta$ after a latency $\mu(\delta)$. Our first contribution is a regret bound for this distributed framework which enables us to determine the relationship between $\delta$, $\mu(\delta)$, and the mini-batch size $b$ that is sufficient to ensure the optimal $\order{ \sqrt{m} }$ asymptotic regret bound is achieved.

Then, we focus specifically on using gossip protocols for approximate distributed averaging~\cite{randomizedGossip,gossipReview}. Our second contribution is a bound on the number of gossip iterations required to satisfy the relationship between the accuracy $\delta$, latency $\mu(\delta)$, and mini-batch size $b$, so that the optimal regret bound is achieved. In particular, the number of iterations scales as $\order{ \frac{\log(n)}{1 - \lambda_2(P)} }$, where $\lambda_2(P)$ is the second largest eigenvalue of the transition matrix $P$ of a random walk on the network topology over which messages are passed. When the network is well connected (e.g., when it is an expander~\cite{kRegExpanders}), then $\frac{1}{1 - \lambda_2(P)} = \Theta(1)$ as $n \rightarrow \infty$, and so the number of gossip iterations (which is proportional to the time spent communicating) is $\order{ \log(n) }$, which is no more than the time spent communicating per round if an exact distributed averaging protocol were used, up to constant factors.

Finally, we study the proposed approach in the setting of stochastic optimization. Specifically, we derive bounds on the optimality gap as a function of the accuracy $\delta$. We show that the number of rounds required to drive the optimality gap below $\epsilon$ (i.e., $\Delta_n(T) \le \epsilon$) is of the order $T = \order{ \frac{1}{n \epsilon^2} }$. If gossip is used for approximate distributed averaging, and if the network is an expander, then the number of gossip iterations per round is also of the order $\order{ \Theta(\log n) }$, and so the speedup relative to using a single processor is nearly linear. This speedup is also observed in our experiments.

\subsection{Paper Organization} 

The remainder of the paper is organized as follows. Section~\ref{sec:prob} introduces notation and states assumptions. Section~\ref{sec:distributedAveraging} discusses exact and approximate protocols for distributed averaging. Section~\ref{sec:results} describes the proposed algorithm and presents our main results. The proofs of all results are given in Section~\ref{sec:analysis}. Results of numerical experiments are presented in Section~\ref{sec:experiments}. We further discuss our results related to other work in the literature in Section~\ref{sec:discussion}, and we conclude in Section~\ref{sec:conc}.

\section{Notation and Assumptions}
\label{sec:prob}

Throughout this paper, for a vector $x \in \R^d$, $\norm{x}$ denotes the Euclidean (i.e., $\ell_2$) norm, $\norm{x} = x^T x$. The set of the first $n$ natural numbers is denoted bys $[n] \defeq \{1,2,\dots,n\}$. The $d$-dimensional vector with all entries equal to $0$ (resp.,~all equal to $1$) is written as $\mathbf{0}_d$ (resp.~$\mathbf{1}_d$).

We focus on the two problem settings described in Sections~\ref{subsec:onlinePrediction} and~\ref{subsec:stochasticOptimization} above. To re-iterate, the samples $x_i(t,s) \in \mathcal{X}$ arriving to the network are assume to be independent and identically distributed, and they are drawn from an unknown distribution over $\mathcal{X}$, where $\mathcal{X} \subseteq \R^d$ is a $d$-dimensional Euclidean space for some $d \ge 1$. The objective is to minimize the regret $R_n(m)$ accumulated after $m$ samples have been processed, in total, across the network. Our results and analysis below make use of the following assumptions.

\begin{ass} \label{ass:constraintSet}
The set $\constraintset$ is a closed, bounded, convex subset of $\R^d$ with diameter $D = \max_{v, w \in \constraintset} \|v - w\|$.
\end{ass}

\begin{ass} \label{ass:lossFunction} The loss function $f(w,x)$ has the following properties.
\renewcommand{\labelenumi}{\alph{enumi}.}
\begin{enumerate}
\item (Convexity) For all $x \in \mathcal{X}$, the loss function $f(w,x)$ is convex in $w$.
\item (Lipschitz continuity) There exists a constant $L \ge 0$ such that for all $x \in \mathcal{X}$ and $w_1, w_2 \in \constraintset$,
\[
\abs{f(w_1, x) - f(w_2, x)} \le L \norm{w_1 - w_2}.
\]

\item (Lipschitz continuous gradients) The loss function $f(w,x)$ is continuously differentiable with respect to $w$ for all $x \in \mathcal{X}$, and there exists a constant $K \ge 0$ such that for all $x \in \mathcal{X}$ and $w_1, w_2 \in \constraintset$,
\[
\norm{\gradient f(w_1, x) - \gradient f(w_2, x)} \le K \norm{w_1 - w_2}.
\]
\end{enumerate}
\end{ass}

Note that Assumptions~\ref{ass:constraintSet} and~\ref{ass:lossFunction} together imply that $f(w, x)$ also has bounded gradients; i.e., $\norm{\gradient f(w, x)} \le L$ for all  $w \in \constraintset$ and all $x\in \mathcal{X}$.

Let $F(w) = \E_x[f(w,x)]$. In each round, the algorithms we consider will make use of gradients $\gradient f(w, x)$. These can be viewed as stochastic or noisy versions of the gradient $\gradient F(w)$. The assumptions stated above (in particular, Lipschitz continuity of $f$ and $\gradient f$) on the loss function are sufficient to ensure that the gradient operator commutes with the expectation~\cite{RockafellarWetts}, and thus $\gradient f(w, x)$ is an unbiased estimate for $\gradient F(w)$. In addition, we assume that the gradients have bounded variance. In particular:

\begin{ass} \label{ass:boundedGradientVariance}
There exists a constant $\sigma^2 \in (0, \infty)$ such that
\[
\E_x\left[ \norm{\gradient f(w,x) - \gradient F(w)}^2 \right] \le \sigma^2.
\]
\end{ass}

After receiving data and suffering the loss in one round, node $i$ may send/receive messages to/from other nodes before updating its predictor, $w_i(t+1)$, at the beginning of the next round. Each node is restricted to only communicate with its neighbours in a graph $G = (V,E)$ where $V = [n]$. Throughout this work we assume that:

\begin{ass} \label{ass:graphConnected}
The graph $G$ is undirected and connected. 
\end{ass}

Of course, $G$ could be the complete graph (i.e., each node communicates with all the others). However, then the amount of communication required per node per round may not scale well with the size of the network, and so in general each node may only be neighbours with a few other nodes in the network.

\section{Distributed Averaging Protocols}
\label{sec:distributedAveraging}

Before proceeding to the description of the algorithmic framework considered in this paper, we discuss distributed averaging protocols, which will be used as a mechanism to coordinate and synchronize the decision variables at different nodes.

Consider a connected network of $n$ nodes where each node $i$ holds a vector $y_i \in \R^d$. The goal of a distributed averaging protocol is for all nodes in the network to compute the average vector, $\ybar = \frac{1}{n} \sum_{i=1}^n y_i$. Let 
\[
y_i^+ = \DistributedAveraging(y_i, i)
\]
denote a function called at node $i$ to invoke a distributed averaging protocol. The first argument is the vector, $y_i \in \R^d$, passed to the routine for averaging, and the second argument indicates this is the instance invoked at node $i$. The function returns a vector $y_i^+ \in \R^d$ as the output at node $i$.

In general, the output $y_i^+$ may not be identical at all nodes, and the quality of a distributed averaging protocol is measured through two metrics: the \emph{accuracy} of the protocol and the \emph{latency} of the protocol. A distributed averaging protocol is said to be have accuracy $\delta \ge 0$ if it guarantees that, for all $i$,
\[
\norm{ y_i^+ - \overline{y} } \le \delta
\]
when the protocol returns. In the context of distributed online prediction, the latency of a distributed averaging protocol is measured with respect to the rate at which observations arrive to the network (and hence, the rate at which predictions must be made). A distributed averaging protocol is said to have latency $\mu > 0$ if a total of $\mu$ observations arrive to the network (i.e., $\mu / n$ observations at each node) during the time it takes for the protocol to return a value $y_i^+$ at each node. A $(\delta, \mu)$-distributed averaging protocol is one which guarantees that accuracy $\delta$ is achieved across the network while requiring latency at most $\mu$ to complete.

Typically the latency of a distributed averaging protocol is related to the size and structure of the network, as well as the accuracy required. In particular, we will be interested in how the latency scales both as a function of the accuracy required and as a function of the size of the network. Next we give examples of exact and approximate distributed averaging protocols.

\subsection{Exact Distributed Averaging with \AllReduce }
One way to exactly achieve distributed averaging in a network is to pass messages along a rooted spanning tree. Starting from the leaves of the spanning tree and assuming that all nodes know the size $n$ of the network, nodes pass a message containing $y_i / n$ to their parents in the tree. Once a node has received messages from all of its children, it computes the sum of the values received in these messages and also adds its own value, normalized by the size of the network, and then it passes a new message with this intermediate aggregate to its parent. This process repeats until the root node has received messages from all of its children. After the root sums the values received from its children together with its own (appropriately normalized) value, it holds the average $\ybar$. Then $\ybar$ can be broadcast down the spanning tree to all nodes, and the protocol terminates when all nodes have received $\ybar$. 

The operation just described corresponds to the \AllReduce primitive which is standard in existing message passing libraries for distributed computing~\cite{mpi}. It is clear that \AllReduce achieves an the accuracy of $\delta = 0$, since all nodes hold the exact average when the protocol terminates. The latency of \AllReduce is proportional to the depth of the spanning tree and so, assuming $G$ admits a balanced spanning tree where the branching factor that is bounded independent of $n$, the latency is of the order $\mu = \Theta( \log n )$.

Exact distributed averaging using \AllReduce is an attractive option when nodes communicate over a network which is stable and reliable. However, when nodes or links may fail, or when the communication channels between neighbouring nodes are unreliable, then there may be considerable overhead involved in establishing and maintaining a spanning tree across all nodes. In such a setting, it may be more desirable to use an approximate distributed averaging protocol with less overhead.

\subsection{Approximate Distributed Averaging with Gossip Protocols }
\label{subsec:gossip}

Gossip protocols are a family of simple message passing schemes that can be used for approximate distributed averaging. Here we focus our discussion on synchronous gossip protocols, but asynchronous versions exist and are equally usable within the algorithmic framework considered in this paper~\cite{randomizedGossip,PushSum,Xiao04}; see~\cite{saberReview,gossipReview} for surveys and see~\cite{TsianosAllertonInvited} for a discussion of practical issues related to implementing asynchronous gossip protocols.

In a gossip protocol, rather than aggregating and broadcasting along a spanning tree, nodes iteratively pass messages to their neighbours in the communication graph $G$. In one iteration of a synchronous gossip algorithm, node $i$ transmits a message to its neighbours containing its current value $y_i$, and it receives messages from each of its neighbours containing their values $y_j$. Then node $i$ updates its value to be a convex combination of its previous values and those received from its neighbours. Specifically, let $P \in [0,1]^{n \times n}$ denote a $n \times n$ matrix of update weights. We require that $P$ respects the topology of the network in the sense that $P_{i,j} > 0$ if and only if either $i = j$ or nodes $i$ and $j$ are neighbours in $G$ (i.e., $(i,j) \in E$). After receiving messages from all of its neighbours, node $i$ updates its value to
\[
y_i \leftarrow \sum_{j=1}^n P_{i,j} y_j \;.
\]
Note that since $P_{i,j} = 0$ if $(i,j) \notin E$ and $i \ne j$, the only non-zero terms in the summation are those corresponding to neighbours of $i$ or to $i$ itself, and so node $i$ receives all the information it needs to compute the update via messages from its immediate neighbours. Let $y_i^{(0)}$ denote the initial vector at node $i$, and let
\[
y_i^{(k)} = \Gossip^k(y_i^{(0)}, i)
\]
denote the vector output at node $i$ after $k$ gossip iterations. Then
\[
y_i^{(k)} = \sum_{j=1}^n [P^k]_{i,j} y_j^{(0)} \;.
\]

In addition to imposing that the sparsity structure of $P$ matches the communication topology $G$, we assume that:
\begin{ass} \label{ass:gossipWeights}
The matrix $P$ is doubly stochastic; i.e., $\sum_{j=1}^n P_{i,j} = 1$ and $\sum_{i=1}^n P_{i,j} = 1$.
\end{ass}
This assumption, combined with the assumption that $G$ is connected (Assumption~\ref{ass:graphConnected}) implies, via the Perron-Frobenius Theorem~\cite{SenetaBook}, that $\lim_{k \rightarrow \infty} P^k = \frac{1}{n} \mathbf{1}_n \mathbf{1}_n^T$. Clearly this is both necessary and sufficient to ensure that $y_i^{(k)} \rightarrow \ybar$ at every node $i$. Moreover, there are well-known bounds on the rate of convergence of these iterations. Since each iteration of gossip can be viewed as taking one discrete-time step of a diffusion over the network, it it not surprising that the rate of convergence depends on the communication topology through the second largest eigenvalue $\lambda_2(P)$ of the matrix $P$.

\begin{lem} \label{lem:gossipRates}
Let Assumptions~\ref{ass:graphConnected} and~\ref{ass:gossipWeights} hold, let $\delta > 0$ be a given scalar, and let $y_i^{(k)} = \Gossip^k(y_i^{(0)}, i)$ denote the output at node $i$ after running $k$ iterations of gossip for distributed averaging. Then it holds that
\[
\norm{ y_i^{(k)} - \ybar } \le \delta
\]
for all $i \in [n]$ if the number of gossip iterations $k$ satisfies
\begin{equation}
k \ge \frac{\log\left( \frac{1}{\delta} \cdot 2 \sqrt{n} \cdot \max_j \norm{ y_j^{(0)} - \overline{y} }\right)}{1 - \lambda_2(P)} \;. \label{eqn:gossipIterations}
\end{equation}
\end{lem}

The proof is given in Appendix~\ref{app:gossipRatesProof}. The latency of gossip-based distributed averaging is proportional to the number of gossip iterations required, and it is evident from Lemma~\ref{lem:gossipRates} that this depends on the size and structure of the network, both directly through $\sqrt{n}$ as well as through $\lambda_2(P)$. It is well-known that the inverse spectral gap, $1/ (1 - \lambda_2(P))$ is related to the structure of the network~\cite{SpectralGraphTheoryBook}. This relationship is especially important for understanding how the performance of gossip protocols (as well as the distributed optimization method introduced below) scale as a function of the network size and structure.

Of particular interest is the case when $G$ is taken from the family of constant-degree expander graphs~\cite{kRegExpanders}. Then $1 - \lambda_2(P) = \Theta(1)$ as $n \rightarrow \infty$, and in addition, every node has the same number of neighbours. For gossip-based distributed averaging run over an expander graph or a complete graph, the resulting latency to obtain accuracy $\delta$ is
\[
\mu = \order{ \log\left(\frac{1}{\delta} \sqrt{n} \max_j \norm{ y_j^{(0)} - \ybar } \right) } \;.
\]
In particular, if both $\delta$ and $\max_j \norm{ y_j^{(0)} - \ybar }$ are fixed constant as $n \rightarrow \infty$, then the latency scales as $\order{ \log(n) }$ which is exactly the same as for exact distributed averaging using $\AllReduce$.

In general, we will see that one can use any connected graph $G$ satisfying the assumptions mentioned above. The expanders mentioned above are especially interesting because of the scaling properties mentioned. For less well-connected graphs, such as rings or two-dimensional grids, the inverse spectral gap $\frac{1}{1 - \lambda_2(P)}$ typically grows as a polynomial in $n$, and so more gossip iterations (and hence, more communications and a higher latency) are required to guarantee an accuracy of $\delta$.

\section{Distributed Dual Averaging with Approximate Mini-Batch Gradients}
\label{sec:alg}

Next we present the general template for the distributed dual averaging algorithm using an approximate distributed averaging protocol to compute the mini-batch gradients. Node $i$ maintains a primal variable $w_i(t) \in \R^d$ and a dual variable $z_i(t) \in \R^d$ which are both initialized to $w_i(1) = z_i(1) = \mathbf{0}_d$. In round $t \ge 1$, after receiving the samples $x_i(t, 1), \dots, x_i(t, \frac{b}{n})$ and suffering the corresponding loss $\sum_{s=1}^{b/n} f(w_i(t), x_i(t,s)$, node $i$ computes the local average gradient
\begin{equation}
g_i(t) = \frac{n}{b} \sum_{s=1}^{\frac{b}{n}} \gradient f(w_i(t), x_i(t,s)) \;. \label{eqn:g_i}
\end{equation}
Then, the nodes update their dual variables by running an approximate distributed averaging protocol,
\begin{equation}
z_i(t+1) = \DistributedAveraging( z_i(t) + g_i(t), i) \;. \label{eqn:z_i}
\end{equation}
Finally, each node updates its primal variable via the proximal projection,
\begin{equation}
w_i(t+1) = \arg\min_{w \in \constraintset} \left\{ \langle w, z_i(t+1) \rangle + \beta(t+1) h(w) \right\} \;, \label{eqn:w_i}
\end{equation}
where $\{ \beta(t) \}_{t=2}^\infty$ is a positive non-decreasing sequence of algorithm parameters and $h(w)$ is a $1$-strongly convex proximal function; i.e., $h(0) = 0$, and for all $\theta \in [0,1]$ and all $w, w' \in \constraintset$,
\[
h(\theta w + (1 - \theta) w') \le \theta \cdot h(w) + (1 - \theta) h(w') - \frac{1}{2} \theta (1 - \theta) \norm{w - w'}^2 \;.
\]
When working in Euclidean domains, as considered in this paper, the typical choice is to take $h(w) = \norm{ w }^2$; see~\cite{nesterovDualOpt,logRegRepeatedGames} for examples of other possible proximal functions that are more suitable in different domains.

For the framework just described, the same predictor $w_i(t)$ is used for all samples $x_i(t,s)$ processed at node $i$ in round $t$. Consequently, the expression for the regret simplifies from~\eqref{eqn:R_nlong} to
\begin{equation}
R_n(m) = \sum_{i=1}^n \sum_{t=1}^{\frac{m}{b + \mu}} \sum_{s=1}^{\frac{b + \mu}{n}} [f(w_i(t), x_i(t,s)) - f(w^*, x_i(t,s))] \;. \label{eqn:R_n}
\end{equation}

Note that, in the update~\eqref{eqn:z_i}, the argument passed to the distributed averaging protocol by node $i$ is $z_i(t) + g_i(t)$, the sum of the dual variable and mini-batch gradient at node $i$. If an approximate distributed averaging protocol is used, then the output at node $i$ is roughly
\begin{align}
z_i(t+1) &= \frac{1}{n} \sum_{i=1}^n \big( z_i(t) + g_i(t) \big) + \xi_i(t) \\
&= \overline{z}(t) + \frac{1}{b} \sum_{i=1}^n \sum_{s=1}^{\frac{b}{n}} \gradient f(w_i(t), x_i(t,s)) + \xi_i(t)\;, \label{eqn:intuition}
\end{align}
where $\xi_i(t)$ is an error term at node $i$ and $\overline{z}(t)$ is the network-average dual variable at round $t$. If an approximate distributed averaging protocol, such as gossip, is used in~\eqref{eqn:z_i} then an upper bound on the magnitude of the error $\xi_i(t)$ is guaranteed and can be made arbitrarily small, typically at the cost of increased communications. Thus, each node effectively updates the dual variable by adding the mini-batch gradient using the $b$ gradient samples obtained across the network.

Before proceeding to our main results, we remark that other distributed versions of the dual averaging algorithm using approximate distributed averaging have proposed in the literature~\cite{dualAveragingTAC}. There are two differences between the algorithm considered in~\cite{dualAveragingTAC} and the framework described above. The first is that the algorithm of~\cite{dualAveragingTAC} only allows for a single gossip iteration in each round, whereas the framework above allows for~\eqref{eqn:z_i} to involve multiple gossip iterations. The second, and more subtle point is that the algorithm of~\cite{dualAveragingTAC} only averages the dual variables $z_i(t)$, which would be equivalent to having node $i$ call
\[
z_i(t+1) = \DistributedAveraging(z_i(t), i) + g_i(t) \;,
\]
from which it is evident that the $z_i(t)$ no longer benefits from the same mini-batch gradient updates as in~\ref{eqn:intuition}; rather, averaging of the mini-batch of gradients $g_i(t)$ from round $t$ is delayed by a round and only appears in $z_i(t+2)$.

\section{Main Results}
\label{sec:results}

\subsection{Online Prediction}
\subsubsection{Regret Bound for General Approximate Distributed Averaging Protocol}

Our first main result demonstrates that the optimal $\order{ \sqrt{m} }$ regret bound can be achieved despite using approximate distributed averaging protocols in the distributed mini-batch framework.

\begin{thm} \label{thm:generalConvergence}
Let Assumptions~\ref{ass:constraintSet}--\ref{ass:graphConnected} hold. Let $z_i(t)$ and $w_i(t)$ be the values generated at node $i$ by the algorithm~\eqref{eqn:g_i}--\eqref{eqn:w_i} described above for $t \ge 1$, and with $\beta(t) = K + \sqrt{\frac{t}{b + \mu}}$. Suppose that the distributed averaging protocol used in~\eqref{eqn:z_i} has accuracy $\delta$ and latency $\mu$. Then the expected regret after processing $m$ samples in total across the network satisfies the bound
\begin{align*}
\E[R_n(m)] &\le (b+\mu) [ F(\wbar(1)) - F(w^*) + K h(w^*)] + \frac{3}{4} \delta^2 K^2 (b+\mu)^{5/2} \\
&\quad + \left[ 2 \sigma^2 \frac{b+\mu}{b} + 2 \delta K D \frac{b+\mu}{n} + 2 \delta L (b + \mu) \right] \sqrt{m} \;.
\end{align*}
\end{thm}

The proof is given in Section~\ref{sec:generalConvergenceProof}. The result above bounds the expected regret as a function of the number of samples $m$ processed by the network, the mini-batch size $b$, the accuracy $\delta$ and latency $\mu$ of the distributed averaging protocol, and other problem-dependent constants. If the accuracy of the distributed averaging protocol is sufficiently small relative to its latency, then we obtain the following simplified bound.

\begin{cor} \label{cor:accuracy-delay}
If $\delta \le \frac{1}{b + \mu}$ then $\E[R_n(m)] = \order{ b + \sqrt{m} }$. Moreover, if $b = m^{\rho}$ for a constant $\rho \in (0, 1/2)$, then $\E[R_n(m)] = \order{ \sqrt{m} }$.
\end{cor}

From the corollary, we see that to obtain the optimal regret scaling we must use a distributed averaging protocol for which the accuracy and latency satisfy the relationship $\delta \le \frac{1}{b + \mu}$. Then, the optimal rate of $\order{ \sqrt{m} }$ follows by scaling $b$ appropriately with respect to the total number of data samples $m$ are to be processed.

\subsubsection{Approximate Mini-Batch Gradients via Gossip-based Distributed Averaging}

Next we investigate conditions under which gossip-based distributed averaging achieves the optimal rates in terms of expected regret. Recall that the gossip protocol for distributed averaging described in Section~\ref{subsec:gossip} is guaranteed to provide $\delta$-accurate output at each node, in the sense that $\norm{ y_j^{(k)} - \ybar } \le \delta$ for all nodes $j \in [n]$, if the number of iterations $k$ satisfies
\[
k \ge \frac{\log\left(\frac{1}{\delta} \cdot 2 \sqrt{n} \cdot \max_j \norm{y_j^{(0)} - \ybar} \right)}{1 - \lambda_2(P)} \;,
\]
where $y_j^{(0)} \in \R^d$ is the initial vector at node $j$, and $\ybar$ is the average of the initial vectors. 

The latency of gossip is proportional to the number of iterations executed; hence, let us suppose that there exists a positive integer $\gamma$ such that $\mu = \gamma k$; i.e., $\gamma$ is the number of samples arriving to the system during a single gossip iteration.

From Corollary~\ref{cor:accuracy-delay} above, we know that the optimal regret bound is obtained if the accuracy and delay satisfy the relationship $\delta \le \frac{1}{b + \mu} = \frac{1}{b + \gamma k}$. Intuitively we expect that it should be feasible to achieve such a bound using gossip iterations on any graph since the accuracy $\delta$ decays exponentially in $k$ (see Lemma~\ref{lem:gossipRates}) whereas the delay $\mu$ is linear in $k$. This intuition is made precise in the following theorem.

\begin{thm} \label{thm:gossipRegretBound}
Let Assumptions~\ref{ass:constraintSet}--\ref{ass:gossipWeights} hold. Let $z_i(t)$ and $w_i(t)$ denote the sequence of values generated at node $i$ by the algorithm~\eqref{eqn:g_i}--\eqref{eqn:w_i} described above for $t \ge 1$. Suppose that the distributed averaging operation in step~\eqref{eqn:z_i} is implemented by running the gossip protocol described in Section~\ref{subsec:gossip} for $k > 0$ iterations; i.e., for step~\eqref{eqn:z_i} take
\[
z_i(t+1) = \Gossip^k(z_i(t) + g_i(t), i) \;.
\]
Suppose, in addition, that there exists a constant $\gamma > 0$ independent of $k$ such that the latency of the gossip protocol is $\mu = \gamma k$, and assume that $\gamma < b$. Take the algorithm parameters to be $\beta(t) = K + \sqrt{\frac{t}{b + \gamma k}}$. If the number of gossip iterations $k$ executed at each round of the algorithm satisfies
\begin{equation}
k \ge \left\lceil \left( \frac{\log(4 L b \sqrt{n} ) + \log(\frac{1}{1 - \lambda_2(P)})}{1 - \lambda_2(P)} + \frac{1}{2 L b} + 1 \right) \frac{1}{1 - \frac{\gamma}{b}} \right\rceil \;, \label{eqn:kthm}
\end{equation}
then, for all $j \in [n]$ and all $t \ge 0$, the values $z_j(t+1)$ satisfy
\[
\norm{ z_j(t+1) - \zbar(t+1) } \le \frac{1}{b + \mu} \;,
\]
where $\zbar(t+1) = \frac{1}{n} \sum_{j=1}^n z_j(t+1)$. Thus, the expected regret is bounded as
\[
\E[R_n(m)] = \order{ b + \sqrt{m} } \;.
\]
Moreover, if the mini-batch size is chosen so that $b = m^\rho$ for some $\rho \in (0, 1/2)$ then $\E[R_n(m)] = \order{ \sqrt{m} }$.
\end{thm}

The proof is given in Section~\ref{sec:gossipRegretBoundProof}. The assumption that $\gamma < b$ simply states that the number of samples $\gamma$ arriving to the network per gossip iteration is less than the total size of the mini-batch. If this is not the case, then the number of samples being lost during the communication phase of the network is larger than the number of samples used in the algorithm updates~\eqref{eqn:g_i}--\eqref{eqn:w_i}, which is not very sensible. For cases of practical interest, one expects that the total latency, $\mu$, and hence $\gamma$, will be much smaller than $b$ (e.g., $\gamma$ remains constant as $m$ grows, whereas $b$ grows as a function of $m$)~\cite{DekelMiniBatches}.

Theorem~\ref{thm:gossipRegretBound} provides a condition on the number of gossip iterations $k$ executed in each round in order to guarantee that the optimal expected regret bound is achieved. From the perspective of scaling laws (as the network size $n$ grows), the three parameters which may depend on $n$ are the inverse spectral gap $1 / (1 - \lambda_2(P))$, the proportionality constant $\gamma$, and the mini-batch size $b$. The way in which the inverse spectral gap depends on $n$ is related to structural properties of the network topology, as discussed in Section~\ref{subsec:gossip} above. 

The proportionality constant $\gamma$ is the number of samples being processed across the network during each gossip iteration (i.e., the latency of each iteration), and it is related to the amount of time it takes to communicate for each gossip iteration. This depends on both the type of communication mechanism available and possibly the structure of the network topology. For example, if nodes communicate using point-to-point messages, then in each iteration, every node must send and receive a message from each of its neighbours, and thus $\gamma$ is proportional to the maximum degree of the network. Hence, $\gamma = \order{ n }$; in the most extreme case, each node sends a message to every other node in each iteration of gossip, and so $\gamma \propto n$. In more practical cases we expect that $\gamma = o(n)$; for example, if the maximum degree in the network grows logarithmically, as it does with high probability in many common random graph models (including the Erd\H{o}s-Renyi model and the random geometric graph model), then $\gamma = \order{ \log n}$. If the degree of each node remains fixed as $n$ grows, as in the constant-degree expander model~\cite{kRegExpanders}), then $\gamma = \order{ 1 }$. It also holds that $\gamma = o(n)$ if nodes communicate using local broadcasts (e.g., over a reliable wireless channel).

The mini-batch size $b$ is linearly proportional to $n$ since $b$ is the total number of samples processed across the network in each round (and so the number of samples per node is $b/n$). Since $\gamma$ only enters the expression~\eqref{eqn:kthm} via the ratio $\gamma/b$, and since $\gamma = \order{ n}$, we see that the term $1 / (1 - \frac{\gamma}{b}) = \order{1}$; when $\gamma = \Theta(n)$, then $1 / (1 - \frac{\gamma}{b})$ remains a constant less than one, and when $\gamma = o(n)$ then $1 / (1 - \frac{\gamma}{b})$ converges to $1$ as $n \rightarrow \infty$. Thus, although the particular communication mechanism affects the constants (through $\gamma$), it does not affect the overall rate at which $k$ grows as a function of $n$.

Fixing a particular class of network topology determines the relevant laws for $1/(1-\lambda_2(P))$, $\gamma$, and $b$, and the communication scaling law (how the amount of communication per round, $k$, grows as a function of $n$) can be determined. Again, of particular interest is the case where $G$ is a constant-degree expander graph. Then $\gamma = o(n)$ because the maximum degree of the graph is a constant independent of $n$, and $1 / (1 - \lambda_2(P)) = \Theta(1)$. Thus, for $G$ being a constant-degree expander graph, $k = \Theta(\log n)$ iterations of gossip suffices.

\subsection{Stochastic Optimization}

\subsubsection{General Optimization Error Bound}

Next, we turn to the problem of stochastic optimization. Recall that in distributed stochastic optimization, a network of $n$ nodes aims to cooperatively solve the problem
\begin{align*}
\text{minimize} &\quad F(w) \defeq \E_x[f(w,x)] \\
\text{subject to} &\quad w \in \constraintset \;.
\end{align*}
We assume a similar framework to that considered above, where $f(w,x)$ and $\constraintset$ satisfy Assumptions~\ref{ass:constraintSet}--\ref{ass:boundedGradientVariance}, and where nodes communicate over a network satisfying Assumption~\ref{ass:graphConnected}.

Unlike distributed online prediction, new samples do not arrive while the nodes are communicating, and so and no loss is suffered (in terms of the optimization error) due to latency. On the other hand, since the motivation of using a distributed system to solve the problem is to \emph{quickly} obtain an accurate solution, it is of interest to know how the time to reach an $\epsilon$-accurate solution scales with the number of nodes in the network.

We consider a method which operate in a similar manner to those described in Section~\ref{sec:alg} for distributed online prediction. Node $i$ maintains its own copy of the optimization variable $w_i(t)$, and the values at each node are updated in synchronous rounds. Within each round, nodes may process a mini-batch of samples locally, communicate with their neighbours, and then update their local variables. 

For stochastic optimization, we augment the dual averaging algorithm described in Section~\ref{sec:alg}. After performing the steps~\eqref{eqn:g_i}--\eqref{eqn:w_i}, node $i$ also computes the running average predictor,
\begin{equation}
\what_i(t+1) \defeq \frac{1}{t+1} \sum_{t'=1}^{t+1} w_i(t') \;, \label{eqn:what_i}
\end{equation}
In stochastic optimization, performance is measured in terms of the optimality gap. After the network has completed $T$ rounds, the optimality gap is given by
\[
\Delta_n(T) \defeq \max_{i \in [n]} \; \E[ F(\what_i(T)) ] - F(w^*) \;.
\]

It is well-known that there is a direct connection between the performance of online prediction methods and stochastic optimization in the single-processor setting~\cite{xiaoDA}, where Jensen's inequality gives that 
\begin{align*}
\Delta_1(T) &\le \frac{1}{T} \E[R_1(T)] \;.
\end{align*}
Thus, any online prediction algorithm with expected regret scaling as $\E[R_1(m)] = \order{ \sqrt{m} }$ can be used to obtain an optimization error of $\Delta_1(T) = \order{ \frac{1}{\sqrt{T}} }$. Equivalently, in order to guarantee that $\Delta_1(T) \le \epsilon$ using such an algorithm, we need to run at least $T \propto \frac{1}{\epsilon^2}$ rounds. 
In the distributed setting, we obtain a similar bound for the distributed dual averaging algorithm with approximate mini-batches.\footnote{Note that, unlike in the single-processor setting just discussed, we cannot simply apply Jensen's inequality to the regret bound from Theorem~\ref{thm:generalConvergence} to obtain this result. This is because the different processors produce sequences $\{\what_i(t)\}$ which are not exactly the same due to the use of an approximate distributed averaging protocol, and these errors must be properly accounted for to obtain a bound on the optimization error.}

\begin{thm} \label{thm:stochasticOptimizationBound}
Let Assumptions~\ref{ass:constraintSet}--\ref{ass:graphConnected} hold. Let $\{\what_i(t)\}_{t \ge 1}$ denote the values generated at node $i$ by repeating the iterations~\eqref{eqn:g_i}--\eqref{eqn:w_i} and~\eqref{eqn:what_i} with algorithm parameters $\beta(t) = K + \sqrt{\frac{t}{b}}$. For some positive integer $C \ge 1$, set the mini-batch size to $b = C n$. Assume that the approximate distributed averaging protocol used in~\eqref{eqn:z_i} guarantees accuracy $\delta \le \frac{1}{b}$ in each round. Then after $T$ rounds, we have
\begin{align}
\Delta_n(T) &\le \frac{1}{T} \left[F(0) - F(w^*) + K h(w^*) + \frac{3 K^2 }{4 \sqrt{C n}} \right]  + \frac{1}{\sqrt{T}} \left[ \frac{\sigma^2}{4 \sqrt{C n}} + \frac{2 K D}{n \sqrt{C n}} + \frac{2 L }{\sqrt{C n}} \right]  \;.
\end{align}
\end{thm}

The proof is given in Section~\ref{sec:stochasticOptimizationBoundProof}. The result above shows that, for a fixed network size $n$, the optimality gap after $T$ rounds is of the order $\Delta_n(T) = \order{ \frac{1}{\sqrt{n T}} }$. Thus, for $\epsilon > 0$, in order to guarantee that $\Delta_n(T) \le \epsilon$ we need to perform at least $T \propto \frac{1}{n \epsilon^2}$ rounds. Similar to the serial setting, the number of rounds scales as $\frac{1}{\epsilon^2}$. However, importantly, we see that in the distributed setting, the number of rounds is reduced by a factor of $n$, and so we obtain a linear scale-up in terms of the network size. In order to translate this scaling in the number of rounds $T$ to a speedup in terms of the runtime, we also need to account for the time spent communicating in each round relative to the time spent computing. This is the topic of the next subsection.

\subsubsection{Latency and Scaling}

In distributed stochastic optimization we do not consider latency in the definition of the optimization error $\Delta_n(T)$ because, unlike the online prediction setting where inputs arrive in a streaming manner and predictions must be made in real-time,\footnote{Recall that the motivation for using a distributed system for online prediction was to enable handling higher rates of predictions.} in stochastic optimization no additional loss or error needs to be accrued while the network is communicating. The primary motivation for using a distributed system in the context of stochastic optimization is to achieve accurate results \emph{faster} than with a single node, by exploiting the ability to process inputs in parallel.

Theorem~\ref{thm:stochasticOptimizationBound} above shows that we achieve a linear speedup in the number of iterations required to guarantee an optimization error of no more than $\epsilon$. Specifically, $T \propto \frac{1}{n \epsilon^2}$. However, we should also account for the latency involved in communication in order to model the overall runtime to reach accuracy $\epsilon$. To this end, we adapt the approach of~\cite{TsianosNIPS2012} for modelling runtimes of distributed optimization methods.

Let us assume that time is normalized so that it takes a processor 1 unit of time to process a single sample $x$ (i.e., to evaluate~$\gradient f(w,x)$ once), and let $\tau > 0$ denote the number of time units it takes a node to transmit its dual variable to one neighbour.\footnote{We will disregard the time it takes to compute the proximal projection~\eqref{eqn:w_i}; relative to the time spent aggregating gradients in~\eqref{eqn:g_i} and the time spent communicating in~\eqref{eqn:z_i}, which both depend on the mini-batch size $b$ and network size $n$, the time to compute~\eqref{eqn:w_i} is $\order{1}$. Often, there is a closed form expression for the solution to~\eqref{eqn:w_i}.}  In the serial (single-processor) setting, the dual averaging algorithm must process $T \propto \frac{1}{\epsilon^2}$ samples in order to guarantee an optimality gap of the order $\Delta_1(T) = \order{ \epsilon }$, and so clearly the runtime also scales as $\order{ \frac{1}{\epsilon^2}}$ time units.

Suppose we run a gossip protocol for $k$ iterations to implement the approximate distributed averaging operation~\eqref{eqn:z_i} in every round. In each iteration, each node must send a copy of its dual variable to each of its neighbours, so the time spent communicating per round is $\tau \cdot k \cdot \deg(G)$, where $\deg(G)$ denotes the maximum node degree in $G$. Before communicating, the nodes each process $b/n$ samples, in parallel, to compute $g_i(t)$, so the total time per round is $\frac{b}{n} + \tau \cdot k \cdot \deg(G)$ time units. Combining this with the fact that the distributed algorithm needs to run for $T \propto \frac{1}{n \epsilon^2}$ rounds, we have that the overall runtime is proportional to
\[
\frac{1}{n \epsilon^2} \left(\frac{b}{n} + \tau \cdot k \cdot \deg(G)\right) = \frac{1}{n \epsilon^2} \left(C + \tau \cdot k \cdot \deg(G)\right) \quad \text{time units,}
\]
where $k$ is the number of gossip iterations required to guarantee $\delta \le \frac{1}{b}$ in every round. 

To determine the appropriate value of $k$ we can directly apply Lemma~\ref{lem:gossipRates} in this case to find that $\delta \le \frac{1}{b}$ if we run at least
\[
k \ge \frac{\log\left(b \cdot 2\sqrt{n} \cdot \max_j \norm{y_j^{(0)} - \ybar}\right)}{1 - \lambda_2(P)}
\]
gossip iterations, where $y_j^{(0)} = z_j(t) + g_j(t)$ is the initial value at each node when the gossip protocol is invoked to compute $z_j(t+1)$. Note that
\begin{align*}
\norm{z_j(t) + g_j(t) - (\zbar(t) + \gbar(t))} &\le \norm{z_j(t) - \zbar(t)} + \norm{g_j(t) - \gbar(t)} \\
&\le \norm{z_j(t) - \zbar(t)} + 2 L \;,
\end{align*}
where the second inequality follows after a second application of the triangle inequality along with recalling that Assumptions~\ref{ass:constraintSet} and~\ref{ass:lossFunction} imply that the gradients $g_j(t)$ are bounded (and hence $\gbar(t)$ is too). In addition, we claim that $\norm{z_j(t) - \zbar(t)} \le \frac{1}{b}$ for all $t$. Clearly this holds for the first round since $z_j(1) = \mathbf{0}_d$ at every node. Proceeding inductively, suppose that $\norm{z_j(t) - \zbar(t)} \le \frac{1}{b}$. Then choosing
\[
k \ge \frac{\log\left(b \cdot 2 \sqrt{n} \cdot (\frac{1}{b} + 2 L)\right)}{1 - \lambda_2(P)} \propto \frac{\log(n)}{1 - \lambda_2(P)}
\]
ensures that $\norm{z_j(t+1) - \zbar(t+1)} \le \frac{1}{b}$, and so it holds for all $t \ge 1$.

Returning to the runtime scaling, we have that overall runtime when gossip is used for approximate distributed averaging is proportional to
\[
\frac{1}{n \epsilon^2} \left(C +  \frac{\tau \log(n)}{1 - \lambda_2(P)} \cdot \deg(G)\right) \quad \text{time units.}
\]
When $G$ is taken from the family of constant-degree expander graphs then $\deg(G) = \Theta(1)$ and $\frac{1}{1 - \lambda_2(P)} = \Theta(1)$, and so the overall runtime scales as
\[
\frac{1}{n \epsilon^2} \left(C + \tau \log(n) \right) \quad \text{time units.}
\]
Thus, the runtime scales nearly linearly, as $\order{\frac{\log n}{n}}$, when communication is over an expander graph. Note that this is identical to the scaling achieved when exact distributed averaging is used, since aggregating along a spanning tree also takes $\order{ \log n }$ time units.

\section{Analysis}
\label{sec:analysis}

\subsection{Proof of Theorem~\ref{thm:generalConvergence}}
\label{sec:generalConvergenceProof}

Let
\[
X(t) \defeq \left\{ x_i(t', s) \colon 1 \le i  \le n, \, 1 \le t' \le t \text{ and } 1 \le s \le \frac{b+\mu}{n} \right\}
\] 
denote the set of all data processed across the network during the first $t$ rounds. Observe that the predictor $w_i(t)$ computed using~\eqref{eqn:w_i} is independent of $x_i(t,s)$ given $X(t-1)$. Recall that $\E[\cdot]$ denotes the expectation with respect to the entire data sequence $X(\frac{m}{b+\mu})$, i.e., with respect to the entire collection of $m$ samples processed by the network when computing $\E[R_n(m)]$. Conditioning on $X(t-1)$ gives
\begin{align*}
\E[f(w_i(t), x_i(t,s))] &= \E\big[ \E[ f(w_i(t), x_i(t,s)) | X(t-1) ] \big] \\
&= \E[F(w_i(t))].
\end{align*}
Also, clearly $\E[f(w^*, x_i(t, s))] = \E[F(w^*)]$. Therefore, for $t \le \frac{m}{b + \mu}$,
\begin{equation}
\E[R_n(m)] = \sum_{t=1}^{\frac{m}{b+\mu}} \sum_{i=1}^n \sum_{s=1}^{\frac{b+\mu}{n}} [\E[F(w_i(t)) - F(w^*)] \;, \label{eqn:expectedRegret}
\end{equation}
and we can focus on bounding the performance with respect to the expected loss function, $F(w)$.

Each predictor $w_i(t)$ is a proximal projection of the dual variable $z_i(t)$ as given in~\eqref{eqn:w_i}. In the sequel we will use the fact that this projection is Lipschitz continuous, as given by the following lemma. (See, e.g.,~\cite{dualAveragingTAC} for a proof.)

\begin{lem} \label{lem:proxProjectionLipschitz}
Let $h \colon \R^d \rightarrow \R$ be a strongly convex function and consider the proximal projection
\[
\proxprojection{ z, \beta } \defeq \arg \min_{w \in \constraintset} \left\{ \langle w, z \rangle + \beta h(w) \right\} \;.
\]
For any vectors $z, z' \in \R^d$ and for any positive scalar $\beta > 0$,
\[
\norm{ \proxprojection{ z, \beta } - \proxprojection{ z', \beta } } \le \frac{1}{\beta} \norm{ z - z' } \;.
\]
\end{lem}

By assumption, the distributed averaging algorithm used in \eqref{eqn:z_i} produces an output $z_i(t+1)$ at node $i$ which satisfies $\|z_i(t+1) - \zbar(t+1) \| \le \delta$ for all $i$, where
\begin{equation}
\zbar(t+1) \defeq \frac{1}{n} \sum_{i=1}^n \left( z_i(t) + g_i(t) \right) \;. \label{eqn:zbar}
\end{equation}
If all nodes were to exactly compute $\zbar(t+1)$ and use this in the projection step~\eqref{eqn:w_i} then they would obtain the predictor
\begin{equation}
\wbar(t+1) \defeq \arg\min_{w \in \constraintset} \left\{ \langle w, \zbar(t+1) \rangle + \beta(t+1) h(w) \right\} \;. \label{eqn:wbar}
\end{equation}
To bound the expected regret across the network in the general case where an approximate distributed averaging protocol is used, we relate the expected regret to the regret that would be accumulated using the reference predictor sequence, $\{\wbar(t)\}_{t \ge 1}$, and then we bound the deviation between the performance of the predictors $\{w_i(t)\}_{t \ge 1}$ from that of the reference predictors $\{\wbar(t)\}_{t \ge 1}$.

Adding and subtracting $F(\wbar(t))$ in~\eqref{eqn:expectedRegret} gives
\begin{align*}
\E[R_n(m)] &= \E\left[ \sum_{t=1}^{\frac{m}{b+\mu}} \sum_{i=1}^n \sum_{s=1}^{\frac{b+\mu}{n}} [F(\wbar(t)) - F(w^*) + F(w_i(t)) - F(\wbar(t))] \right] \\
&= \E\left[ (b + \mu) \sum_{t=1}^{\frac{m}{b+\mu}} [F(\wbar(t)) - F(w^*)] + \frac{b+\mu}{n} \sum_{t=1}^{\frac{m}{b+\mu}} \sum_{i=1}^n [F(w_i(t)) - F(\wbar(t))] \right] \;.
\end{align*}
Since $f(w, x)$ is Lipschitz continuous in $w$ for all $x$ by Assumption~\ref{ass:lossFunction}, it follows that $F(w)$ is Lipschitz continuous with the same constant. Using this property along with Lemma~\ref{lem:proxProjectionLipschitz} and linearity of expectation gives
\begin{align}
\E[R_n(m)] &\le \E\left[ (b + \mu) \sum_{t=1}^{\frac{m}{b+\mu}} [F(\wbar(t)) - F(w^*)] + \frac{(b+\mu)L}{n} \sum_{t=1}^{\frac{m}{b+\mu}} \sum_{i=1}^n \norm{ w_i(t) - \wbar(t) } \right] \nonumber \\
&\le (b + \mu) \sum_{t=1}^{\frac{m}{b+\mu}} \E[F(\wbar(t)) - F(w^*)] + \frac{(b+\mu)L}{n} \sum_{t=1}^{\frac{m}{b+\mu}} \sum_{i=1}^n \frac{1}{\beta(t)} \E[ \norm{ z_i(t) - \zbar(t) } ] \nonumber \\
&\le (b + \mu) \sum_{t=1}^{\frac{m}{b+\mu}} \E[F(\wbar(t)) - F(w^*)] + (b+\mu) L \delta \sum_{t=1}^{\frac{m}{b+\mu}} \frac{1}{\beta(t)} \;, \label{eqn:intermediateRegret}
\end{align}
where the last inequality follows based on the assumed accuracy of the approximate distributed averaging protocol. The main work remaining to prove Theorem~\ref{thm:generalConvergence} involves bounding the first sum in~\eqref{eqn:intermediateRegret}. To this end we have:

\begin{lem} \label{lem:averagedRegretBound}
Let the assumptions of Theorem~\ref{thm:generalConvergence} hold, and let $a(t) = \sqrt{\frac{t}{b + \mu}}$. Then
\begin{align*}
\sum_{t=1}^{\frac{m}{b + \mu}} \E[F(\wbar(t)) - F(w^*)] &\le F(\wbar(1)) - F(w^*) + (K + a(\tfrac{m}{b + \mu})) h(w^*) + \frac{\sigma^2}{4b} \sum_{t=1}^{\frac{m}{b+\mu} - 1} \frac{1}{a(t)} \\
&\quad + \frac{\delta^2 K^2}{4} \sum_{t=1}^{\frac{m}{b+\mu} - 1} \frac{1}{a(t) (K + a(t))^2} + \frac{\delta K D}{n} \sum_{t=1}^{\frac{m}{b+\mu} - 1} \frac{1}{K + a(t)} \;.
\end{align*}
\end{lem}

The proof of Lemma~\ref{lem:averagedRegretBound} is given in Section~\ref{sec:averagedRegretBoundProof}. This addresses the first sum in~\eqref{eqn:intermediateRegret}. In order to bound the remaining terms, with $a(t) = \sqrt{\frac{t}{b + \mu}}$ as in the statement of the lemma, note that
\begin{align*}
\sum_{t=1}^{\frac{m}{b+\mu} - 1} \frac{1}{\beta(t)} &= \sum_{t=1}^{\frac{m}{b+\mu} - 1} \frac{1}{K + a(t)} \\
&\le \sum_{t=1}^{\frac{m}{b+\mu} - 1} \frac{1}{a(t)} \\
&\le \sqrt{b + \mu} \left(1 + \int_{1}^{\frac{m}{b+\mu}} t^{-1/2} dt\right) \\
&\le 2 \sqrt{m} \;,
\end{align*}
and 
\begin{align*}
\sum_{t=1}^{\frac{m}{b+\mu} - 1} \frac{1}{a(t) (K + a(t))^2} &\le \sum_{t=1}^{\frac{m}{b+\mu} - 1} \frac{1}{a(t)^3} \\
&= (b+\mu)^{3/2} \sum_{t=1}^{\frac{m}{b+\mu}} t^{-3/2} \\
&\le (b+\mu)^{3/2} \left(1 + \int_1^{\frac{m}{b+\mu}} t^{-3/2} dt \right) \\
&\le 3 (b+\mu)^{3/2} \;.
\end{align*}
Combining these inequalities with Lemma~\ref{lem:averagedRegretBound} and~\eqref{eqn:intermediateRegret} gives
\begin{align*}
\E[R_n(m)] &\le (b + \mu) \left[ F(\wbar(1)) - F(w^*) + \left(K + \frac{\sqrt{m}}{b + \mu}\right) h(w^*) + \frac{2 \sigma^2}{b} \sqrt{m} + \frac{3}{4} \delta^2 K^2 (b + \mu)^{3/2} + \frac{2 \delta K D}{n} \sqrt{m} \right] \\
&\quad + 2 \delta (b + \mu) L \sqrt{m} \\
&= (b + \mu) [F(\wbar(1)) - F(w^*) + K h(w^*)] + \frac{3}{4} \delta^2 K^2 (b + \mu)^{5/2} \\
&\quad + \left[ 2 \sigma^2 \frac{b + \mu}{b} + 2 \delta K D \frac{b + \mu}{n} + 2 \delta L (b + \mu)\right] \sqrt{m} \;,
\end{align*}
which completes the proof of Theorem~\ref{thm:generalConvergence}.

\subsection{Proof of Lemma~\ref{lem:averagedRegretBound}}
\label{sec:averagedRegretBoundProof}

We seek to upper-bound the sum
\[
\sum_{t=1}^{\frac{m}{b + \mu}} \E[F(\wbar(t)) - F(w^*)] \;,
\]
the regret of the averaged sequence of predictors $\{\wbar(t)\}$ after $\frac{m}{b+\mu}$ rounds. From the definition of $\zbar(t+1)$ in \eqref{eqn:zbar} and from~\eqref{eqn:z_i}, we have
\[
\zbar(t+1) = \zbar(t) + \gbar(t) \;,
\]
where
\begin{align*}
\gbar(t) &= \frac{1}{n} \sum_{i=1}^n g_i(t) \\
&= \frac{1}{b} \sum_{i=1}^n \sum_{s=1}^{\frac{b}{n}} \gradient f(w_i(t), x_i(t,s)) \;.
\end{align*}

Note that, for $\wbar(t)$ as defined in~\eqref{eqn:wbar}, $\wbar(t) \ne \frac{1}{n} \sum_{i=1}^n w_i(t)$ since the projection operator is not linear. Note also that the sequence $\wbar(t)$ is not equivalent to running a serial (i.e., single processor) version of dual averaging with mini-batches since the gradients $g_i(t)$ computed at each node (cf.,~\eqref{eqn:g_i}) are evaluated at different points, $w_i(t)$, whereas a serial implementation would use
\[
\ghat(t) = \frac{1}{b} \sum_{i=1}^n \sum_{s=1}^{\frac{b}{n}} \gradient f(\wbar(t), x_i(t,s)) \;,
\]
in place of $\gbar(t)$, with all gradients evaluated at $\wbar(t)$.

There are two sources of errors in the gradients which must be taken into account. The first source stems from the observation just mentioned, that the updates at different nodes involve evaluating gradients at different locations. The second source of errors is due to the randomness of the gradients $\gradient f(w, x)$ calculated using data $x$ rather than directly using the gradients $\gradient F(w)$. Let us define two gradient error vectors,
\begin{align*}
q(t) &\defeq \ghat(t) - \gradient F(\wbar(t)), \text{ and } \\
r(t) &\defeq \gbar(t) - \ghat(t) \;.
\end{align*}
These error vectors satisfy the properties stated in the following lemma that is proved in Appendix~\ref{sec:errorPropertiesProof}.

\begin{lem} \label{lem:errorProperties}
Let Assumptions~\ref{ass:constraintSet}--\ref{ass:boundedGradientVariance} hold. Then for $t = 1,\dots,\tfrac{m}{b+\mu}$,
\begin{align}
\E[\langle q(t), w^* - \wbar(t) \rangle] &= 0, \\
\E[\langle r(t), w^* - \wbar(t) \rangle] &\le \frac{KD}{n \beta(t)} \sum_{i=1}^n \E[ \norm{ z_i(t) - \zbar(t) } ], \\
\E[ \norm{ q(t) }^2] &\le \frac{\sigma^2}{b}, \\
\E[ \norm{ r(t) }^2] &\le \frac{K^2}{n^2 \beta(t)^2} \sum_{i=1}^n \sum_{j=1}^n \E[ \norm{ z_i(t) - \zbar(t) } \norm{ z_j(t) - \zbar(t) } ],
\end{align}
where $\E[\cdot]$ denotes the expectation with respect to the joint distribution over the set of random variables $X(\tfrac{m}{b+\mu})$.
\end{lem}

Let us denote the linear (first-order) model for $F(w)$ evaluated at $\wbar(t)$ by
\[
l_t(w) \defeq F(\wbar(t)) + \langle \gradient F(\wbar(t)), w - \wbar(t) \rangle \;.
\]
We also make use of an approximation to $l_t(w)$ where $\gradient F(\wbar(t))$ is replaced with $\gbar(t)$,
\begin{align}
\lhat_t(w) &\defeq F(\wbar(t)) + \langle \gbar(t), w - \wbar(t) \rangle \nonumber \\
&= l_t(w) + \langle q(t), w - \wbar(t) \rangle + \langle r(t), w - \wbar(t) \rangle \;. \label{eqn:lhat}
\end{align}

From Assumption~\ref{ass:lossFunction}, we obtain that $F(w)$ is convex and has Lipschitz continuous gradients. It follows that (see, e.g., Lemma~1.2.3 in~\cite{NesterovLectures})
\begin{align*}
F(\wbar(t+1)) &\le l_t(\wbar(t+1)) + \frac{K}{2} \norm{ \wbar(t+1) - \wbar(t) } ^2 \\
&= \lhat_t(\wbar(t+1)) - \langle q(t), \wbar(t+1) - \wbar(t) \rangle - \langle r(t), \wbar(t+1) - \wbar(t) \rangle + \frac{K}{2} \norm{ \wbar(t+1) - \wbar(t) }^2 \;.
\end{align*}
Let $a(t) = \beta(t) - K$.
Applying the Cauchy-Schwarz inequality and adding and subtracting $a(t)$ in the term multiplying $\norm{ \wbar(t+1) - \wbar(t) }^2$ gives
\begin{align}
F(\wbar(t+1)) &\le l_t(\wbar(t+1)) + \frac{K + a(t)}{2} \norm{ \wbar(t+1) - \wbar(t) } ^2 \nonumber \\
&\quad + \norm{ q(t) } \, \norm{ \wbar(t+1) - \wbar(t) } - \frac{a(t)}{4} \norm{ \wbar(t+1) - \wbar(t) }^2 \nonumber \\
&\quad + \norm{ r(t) } \, \norm{ \wbar(t+1) - \wbar(t) } - \frac{a(t)}{4} \norm{ \wbar(t+1) - \wbar(t) }^2 \;. \label{eqn:intermediateFwbar}
\end{align}

Observe that
\begin{align*}
\norm{ q(t) } \,& \norm{ \wbar(t+1) - \wbar(t) } - \frac{a(t)}{4} \norm{ \wbar(t+1) - \wbar(t) }^2 \\
&= \frac{\norm{ q(t) }^2}{4 a(t)} - \left[ \frac{\norm{q(t)}}{\sqrt{4 a(t)}} - \sqrt{\frac{a(t)}{4}} \norm{ \wbar(t+1) - \wbar(t) } \right]^2 \\
&\le \frac{\norm{ q(t) }^2}{4 a(t)} \;, 
\end{align*}
and a similar bound holds if $q(t)$ is replaced by $r(t)$.

Using these bounds in~\eqref{eqn:intermediateFwbar} gives
\begin{align}
F(\wbar(t+1)) &\le \lhat_t(\wbar(t+1)) + \frac{K + a(t)}{2} \norm{ \wbar(t+1) - \wbar(t) } ^2 \nonumber \\
&\quad + \frac{\norm{q(t)}^2}{4 a(t)} + \frac{\norm{r(t)}^2}{4 a(t)} \;. \label{eqn:intermediateFwbar2}
\end{align}
Next, we make use of two lemmas to simplify this expression further. The first lemma gives a relation between $\wbar(t)$ and $\lhat_t(\wbar(t))$.

\begin{lem} \label{lem:wbar_lhat}
Let $\zbar(t)$ and $\wbar(t)$ be defined as in~\eqref{eqn:zbar} and~\eqref{eqn:wbar}. Then
\[
\wbar(t) = \arg\min_{w \in \constraintset} \left\{ \sum_{t'=1}^{t-1} \lhat_{t'}(w) + \beta(t) h(w) \right\} \;.
\]
\end{lem}

\begin{proof}
Since $z_i(1) = 0$ by definition, $\zbar(1) = 0$ also, and we have $\zbar(t) = \sum_{t'=1}^{t-1} \gbar(t')$. It follows that
\begin{equation}
\wbar(t) = \arg\min_{w \in \constraintset} \left\{ \left\langle \sum_{t'=1}^{t-1} \gbar(t'), w \right\rangle + \beta(t) h(w) \right\} \;. \label{eqn:wbarIntermediate}
\end{equation}
Now consider the inner product and observe that, by adding and subtracting $F(\wbar(t'))$ and $\ghat(t')$, we have
\begin{align*}
\left\langle \sum_{t'=1}^{t-1} \gbar(t'), w \right\rangle &= \sum_{t'=1}^{t-1} [ \langle \gradient F(\wbar(t')), w \rangle + \langle \ghat(t') - \gradient F(\wbar(t')), w \rangle + \langle \gbar(t') - \ghat(t'), w \rangle ] \;.
\end{align*}
Moreover, since $F(\wbar(t'))$ and $\wbar(t')$ do not depend on $w$, minimizing the objective in~\eqref{eqn:wbarIntermediate} is equivalent to minimizing
\[
\sum_{t'=1}^{t-1}[  F(\wbar(t')) + \langle \gradient F(\wbar(t')), w - \wbar(t') \rangle + \langle \ghat(s) - \gradient F(\wbar(t')), w - \wbar(t') \rangle + \langle \gbar(t') - \ghat(t'), w - \wbar(t') \rangle] + \beta(t) h(w)
\]
over $w$, which is what we wanted to show.
\end{proof}

The second lemma we will make use of is proved in~\cite{TsengAccelerated} and states a widely-used property of strongly convex functions.

\begin{lem} \label{lem:stronglyConvex}
Let $\constraintset$ be a closed convex set, let $\phi(w)$ be a convex function on $\constraintset$, and let $\htilde(w)$ be a $\mu$-strongly convex function defined on $\constraintset$. Define
\[
w^+ \defeq \arg\min_{w \in \constraintset} \{ \phi(w) + \htilde(w) \} \;.
\]
Then for all $w \in \constraintset$,
\[
\phi(w) + \htilde(w) \ge \phi(w^+) + \htilde(w^+) + \frac{\mu}{2} \norm{ w - w^+ }^2 \;.
\]
\end{lem}

Now consider Lemma~\ref{lem:stronglyConvex} with $\phi(w) = \sum_{t'=1}^{t-1} \lhat_{t'}(w)$ and $\htilde(w) = (K + a(t)) h(w)$. Since $h(w)$ is $1$-strongly convex, $\htilde(w)$ is $(K + a(t))$-strongly convex. Thus
\begin{align*}
\frac{K + a(t)}{2} \norm{ \wbar(t+1) - \wbar(t) }^2 &\le \sum_{t'=1}^{t-1} \lhat_{t'}(\wbar(t+1)) + (K + a(t)) h(\wbar(t+1)) \\
&\quad - \sum_{t'=1}^{t-1} \lhat_{t'}(\wbar(t)) - (K + a(t)) h(\wbar(t)) \;.
\end{align*}
Using this bound in~\eqref{eqn:intermediateFwbar2} gives
\begin{align*}
F(\wbar(t+1)) &\le \sum_{t'=1}^t \lhat_{t'}(\wbar(t+1)) + (K + a(t)) h(\wbar(t+1)) \\
&\quad - \sum_{t'=1}^{t-1} \lhat_{t'}(\wbar(t)) - (K + a(t)) h(\wbar(t)) \\
&\quad + \frac{ \norm{ q(t) }^2}{4 a(t)} + \frac{ \norm{ r(t) }^2 }{4 a(t)} \;.
\end{align*}
Summing both sides over $t=1,\dots, \frac{m}{b+\mu} - 1$ leads to many terms on the right-hand side cancelling, and we are left with
\begin{align*}
\sum_{t=2}^{\frac{m}{b+\mu}} F(\wbar(t)) &\le \sum_{t=1}^{\frac{m}{b+\mu} - 1} \lhat_t(\wbar(\tfrac{m}{b+\mu})) + (K + a(t)) h(\wbar(\tfrac{m}{b+\mu})) \\
&\quad + \sum_{t=1}^{\frac{m}{b+\mu} -1} \frac{\norm{ q(t) }^2 + \norm{ r(t) }^2}{4 a(t)}\;.
\end{align*}
From Lemma~\ref{lem:wbar_lhat}, it follows that
\[
\sum_{t=1}^{\frac{m}{b+\mu} - 1} \lhat_t(\wbar(\tfrac{m}{b+\mu})) + (K + a(t)) h(\wbar(\tfrac{m}{b+\mu})) \le \sum_{t=1}^{\frac{m}{b+\mu} - 1} \lhat_t(w^*) + (K + a(\tfrac{m}{b+\mu})) h(w^*) \;,
\]
and so, replacing $\lhat_t(w)$ with the characterization given in~\eqref{eqn:lhat}, we have
\begin{align*}
\sum_{t=2}^{\frac{m}{b+\mu}} F(\wbar(t)) &\le \sum_{t=1}^{\frac{m}{b+\mu} - 1} l_t(w^*) + \sum_{t=1}^{\frac{m}{b+\mu}} [\langle q(t), w^* - \wbar(t) \rangle + \langle r(t), w^* - \wbar(t) \rangle ] \\
&\quad + (K + a(\tfrac{m}{b+\mu})) h(w^*) + \sum_{t=1}^{\frac{m}{b+\mu} - 1} \frac{\norm{q(t)}^2 + \norm{r(t)}^2}{4 a(t)} \\
&\le \left(\frac{m}{b+\mu} - 1\right) F(w^*) + \sum_{t=1}^{\frac{m}{b+\mu} -1} [ \langle q(t), w^* - \wbar(t) \rangle + \langle r(t), w^* - \wbar(t) \rangle ] \\
&\quad + (K + a(\tfrac{m}{b+\mu})) h(w^*) \sum_{t=1}^{\frac{m}{b+\mu}-1} \frac{\norm{q(t)}^2 + \norm{r(t)}^2}{4 a(t)} \;.
\end{align*}

Add $F(\wbar(1)) - F(w^*)$ to both sides, take the expectation with respect to the data $X(\frac{m}{b + \mu})$, and apply the bounds given in Lemma~\ref{lem:errorProperties} to obtain the desired bound,
\begin{align*}
\sum_{t=1}^{\frac{m}{b+\mu}} \E[F(\wbar(t)) - F(w^*)] &\le F(\wbar(1)) - F(w^*) + (K + a(\tfrac{m}{b+\mu})) h(w^*) + \sum_{t=1}^{\frac{m}{b+\mu}-1} \frac{\sigma^2 / b}{4 a(t)} \\
&\quad + \frac{K^2}{4 n^2} \sum_{t=1}^{\frac{m}{b+\mu} -1} \frac{1}{a(t) (K + a(t))^2} \sum_{i=1}^n \sum_{j=1}^n \E[ \norm{z_i(t) - \zbar(t)} \cdot \norm{z_j(t) - \zbar(t)}] \\
&\quad + \frac{KD}{n} \sum_{t=1}^{\frac{m}{b+\mu}-1} \frac{1}{K + a(t)} \E[\norm{z_i(t) - \zbar(t)}] \\
&\le F(\wbar(1)) - F(w^*) + (K + a(\tfrac{m}{b+\mu})) h(w^*) + \sum_{t=1}^{\frac{m}{b+\mu}-1} \frac{\sigma^2 / b}{4 a(t)} \\
&\quad + \frac{\delta^2 K^2}{4} \sum_{t=1}^{\frac{m}{b+\mu} -1} \frac{1}{a(t) (K + a(t))^2} \\
&\quad + \frac{\delta KD}{n} \sum_{t=1}^{\frac{m}{b+\mu}-1} \frac{1}{K + a(t)}  \;.
\end{align*}

\subsection{Proof of Theorem~\ref{thm:gossipRegretBound}} 
\label{sec:gossipRegretBoundProof}

Recall from Lemma~\ref{lem:gossipRates} that, given initial values $y_j^{(0)} \in \R^d$ at each node $j \in [n]$ and given a constant $\delta > 0$, the distance between the average $\ybar = \frac{1}{n} \sum_{j =1}^n y_j^{(0)}$ and the value $y_j^{(k)}$ at each node after $k$ iterations satisfies
\[
\max_j \norm{ y_j^{(k)} - \ybar } \le \delta
\]
if
\begin{equation}
k \ge \frac{\log\left( \frac{1}{\delta} \cdot 2 \sqrt{n} \cdot \max_j \norm{ y_j^{(0)} - \ybar } \right) }{1 - \lambda_2(P)} \;. \label{eqn:kbound}
\end{equation}
We wish to determine a condition on the number of gossip iterations $k$ such that the accuracy $\delta$ and delay $\mu$ satisfy $\delta \le \frac{1}{b + \mu} = \frac{1}{b + \gamma k}$.

When gossip is used for distributed averaging within the algorithm~\eqref{eqn:g_i}--\eqref{eqn:w_i} to compute the values $\{z_j(t+1)\}_{j \in [n]}$ at round $t+1$, we have initial value $y_j^{(0)} = z_j(t) + g_j(t)$ at node $j$, and thus $\ybar = \zbar(t) + \gbar(t)$, where
\[
\zbar(t) = \frac{1}{n} \sum_{j=1}^n z_j(t) \quad \text{ and } \quad \gbar(t) = \frac{1}{n} \sum_{j=1}^n g_j(t) \;.
\]
Note that $\norm{ \gbar(t) } \le L$ since $\norm{ g_j(t) } \le L$ for all $j \in [n]$ by Assumption~\ref{ass:lossFunction}. 

To show that the accuracy achieved is $\delta \le \frac{1}{b + \gamma k}$, we must show that $\norm{ z_j(t) - \zbar(t)} \le \delta$ for all $j \in [n]$ and all $t$. We will show this inductively, while also deriving the required number of gossip iterations $k$. For the induction base, note that the algorithm is initialized with $z_j(1) = 0$ for all $j \in [n]$, and so $\norm{z_j(1) - \zbar(1)} = 0$. For the induction step, assume that $\norm{z_j(t) - \zbar(t)} \le \delta$ for all $j \in [n]$. Then we have that
\begin{align*}
\max_j \norm{ y_j^{(0)} - \ybar } &= \max_j \norm{ \big(z_j(t) + g_j(t)\big) - \big(\zbar(t) - \gbar(t)\big) } \\
&\le \max_j \big\{ \norm{z_j(t) - \zbar(t)} + \norm{ g_j(t) - \gbar(t) } \big\} \\
&\le \max_j \{ \delta + 2 L\} \\
&\le \frac{1}{b + \gamma k} + 2 L \;.
\end{align*}
Using this together with the relation $\frac{1}{\delta} = b + \gamma k$ in~\eqref{eqn:kbound}, our task becomes that of finding the smallest integer $k$ such that
\[
k \ge \frac{\log\left( 2 \sqrt{n} (1 + 2 L b + 2 L \gamma k)) \right)}{1 - \lambda_2(P)} \;.
\]

For $x \in \R$, $x > \frac{- (1 + 2 L b)}{2 L \gamma}$, define the function
\[
\phi(x) \defeq \frac{\log\left( 2 \sqrt{n} (1 + 2 L b + 2 L \gamma x)) \right)}{1 - \lambda_2(P)} \;.
\]
Since $\frac{1}{1 - \lambda_2(P)}$, $\delta$, $L$, $n$, and $b$ are all positive constants by assumption, it follows that $\phi(0) > 0$. Therefore, $\phi(x)$ intersects the diagonal $\psi(x) = x$ at exactly one point $x > 0$. The exact solution to the equation $x = \phi(x)$ is
\[
x = \frac{- 1}{1 - \lambda_2(P)} \cdot W\left(- \frac{(1 - \lambda_2(P))}{4 L \gamma \sqrt{n}} \exp\left\{-\frac{(1 - \lambda_2(P)) (1 + 2 L b) }{ (2 L \gamma)} \right\} \right) - \frac{1 + 2 L b}{2 L \gamma} \;,
\]
where $W(\cdot)$ is the Lambert W-function.\footnote{The function $W(\cdot)$ is defined as the solution to the equation $W(x) e^{W(x)} = x$, and it has a real solution for $x \ge -1/e$; see~\cite{Corless96} for more.} Given $b$, $L$, $n$, $\gamma$, and $\lambda_2(P)$, one could numerically determine an appropriate value for $k$, but because $W(\cdot)$ is transcendental, this is not directly useful for obtaining the closed-form expression we desire.

Instead, since $\phi(x)$ is concave in $x$, it holds that for any $\xhat > \frac{- (1 + 2 L b)}{2 L \gamma}$,
\begin{align*}
\phi(x) &\le \phi(\xhat) + \phi'(\xhat) \cdot (x - \xhat)  \\
&= \frac{\log(2\sqrt{n} ( 1 + 2 L b + 2 L \gamma \xhat ))}{1 - \lambda_2(P)} + \frac{L \gamma (x - \xhat) }{(1 - \lambda_2(P))(1 + 2 L b + 2 L \gamma \xhat )} \;.
\end{align*}
Thus, $x \ge \phi(x)$ if $x \ge x^*$ where
\[
x^* = \left( \frac{\log(2 \sqrt{n} (1 + 2 L b + 2 L \gamma \xhat ))}{1 - \lambda_2(P)} - \frac{2 L \gamma \xhat}{(1 - \lambda_2(P)) (1 + 2 L b + 2 L \gamma \xhat) } \right) \cdot \frac{1}{1 - \frac{2 L \gamma}{(1 - \lambda_2(P)) (1 + 2 L b + 2 L \gamma \xhat)}} \;.
\]
Now, take
\[
\xhat = \frac{1}{2 L \gamma} \left(\frac{2 L b}{1 - \lambda_2(P)} - 1 - 2 L b \right) \;.
\]
Then
\begin{align*}
x^* &= \left( \frac{\log(4 \sqrt{n} L b)}{1 - \lambda_2(P)} + \frac{\log(\frac{1 }{ 1 - \lambda_2(P)} )}{1 - \lambda_2(P)} - \frac{1}{1 - \lambda_2(P)} + \frac{1}{2 L b} + 1 \right) \frac{1}{1 - \frac{\gamma}{b}} \\
&\le \left( \frac{\log(4 L b \sqrt{n} )}{1 - \lambda_2(P)} + \frac{\log(\frac{1}{1 - \lambda_2(P)})}{1 - \lambda_2(P)} + \frac{1}{2 L b} + 1 \right) \frac{1}{1 - \frac{\gamma}{b}} \;.
\end{align*}

Rounding up to the next largest integer gives the value of $k$ in the statement of the theorem. Thus, we have shown that $k^* = \lceil x^* \rceil$ satisfies the relation $k^* \ge \phi(k^*)$, and so 
\[
\max_j \norm{ z_j(t+1) - \zbar(t+1)} \le \frac{1}{b + \gamma k} \le \frac{1}{b + \mu}
\]
holds, completing the induction step. It follows from Corollary~\ref{cor:accuracy-delay} that the expected regret is bounded by $\E[R_n(m)] = \order{ b + \sqrt{m} }$, and so choosing $b = m^\rho$ for some $\rho \in (0, 1/2)$ gives that $\E[R_n(m)] = \order{ \sqrt{m} }$.

\subsection{Proof of Theorem~\ref{thm:stochasticOptimizationBound}}
\label{sec:stochasticOptimizationBoundProof}

Similar to the proof of Theorem~\ref{thm:generalConvergence}, we will bound the optimization error uniformly over all nodes by comparing the performance of the sequence $\{\what_i(t)\}$ to that of an appropriate averaged sequence. Specifically, define
\begin{equation}
\wbarhat(t) \defeq \frac{1}{t} \sum_{t'=1}^t \wbar(t') \;,
\end{equation}
where $\wbar(t')$ is as defined in~\eqref{eqn:wbar}.

Recall that the network collectively processes $b$ samples in each round of the algorithm, so that a total of $m = b T$ samples have been processed after $T$ rounds. We seek an upper bound on $\E[F(\what_i(\frac{m}{b})) - F(w^*)]$. Adding and subtracting $F(\wbarhat(\frac{m}{b}))$, applying Jensen's inequality, and using Lipschitz continuity of $F(\cdot)$ gives
\begin{align}
\E[F(\what_i(\tfrac{m}{b})) - F(w^*)] &= \E[F(\wbarhat(\tfrac{m}{b})) - F(w^*) + F(\what_i(\tfrac{m}{b})) - F(\wbarhat(\tfrac{m}{b}))] \nonumber \\
&\le \E\left[ \frac{b}{m} \sum_{t=1}^{\frac{m}{b}} [F(\wbar(t)) - F(w^*)] + \frac{b}{m} \sum_{t=1}^{\frac{m}{b}} [F(w_i(t)) - F(\wbar(t))] \right] \nonumber \\
&\le \E\left[ \frac{b}{m} \sum_{t=1}^{\frac{m}{b}} [F(\wbar(t)) - F(w^*)] + \frac{b L}{m} \sum_{t=1}^{\frac{m}{b}} \norm{ w_i(t) - \wbar(t) } \right] \;. \label{eqn:tmp0}
\end{align}

Applying Lemma~\ref{lem:averagedRegretBound}, we have that
\begin{align}
\E\left[ \frac{b}{m} \sum_{t=1}^{\frac{m}{b}} [F(\wbar(t)) - F(w^*)] \right] &\le \frac{b}{m} \left[ F(\wbar(1)) - F(w^*) + (K + a(\tfrac{m}{b})) h(w^*) + \frac{\sigma^2}{4 b} \sum_{t=1}^{\frac{m}{b}} \frac{1}{a(t)} \right. \nonumber \\
&\qquad\quad \left. + \frac{\delta^2 K^2}{4} \sum_{t=1}^{\frac{m}{b}} \frac{1}{a(t) (K + a(t))^2} + \frac{\delta K D}{n} \sum_{t=1}^{\frac{m}{b}} \frac{1}{K + a(t)} \right]. \label{eqn:tmp1}
\end{align}
Also, using Lipschitz continuity of the proximal projection operation (cf.~Lemma~\ref{lem:proxProjectionLipschitz}), we have that
\begin{align}
\E\left[ \frac{b L}{m} \sum_{t=1}^{\frac{m}{b}} \norm{ w_i(t) - \wbar(t) } \right] &\le \frac{b L}{m} \sum_{t=1}^{\frac{m}{b}} \frac{1}{\beta(t)} \norm{ z_i(t) - \zbar(t) } \;. \label{eqn:tmp2}
\end{align}
Combining \eqref{eqn:tmp0}, \eqref{eqn:tmp1}, and \eqref{eqn:tmp2}, and recalling that it is assumed in the statement of the theorem that the approximate distributed averaging protocol has accuracy $\norm{ z_i(t) - \zbar(t) } \le \delta \le \frac{1}{b}$, we obtain
\begin{align*}
\E[F(\what_i(\tfrac{m}{b})) - F(w^*)] &\le \frac{1}{T} [F(\wbar(1)) - F(w^*)] + \frac{1}{T} K h(w^*) + \frac{\sigma^2}{4 \sqrt{b T}} + \frac{3 K^2}{4 T \sqrt{b}} + \frac{2 K D}{n \sqrt{b T}} + \frac{2L}{\sqrt{bT}} \;.
\end{align*}
Using the fact that $b = C n$, and rearranging terms gives the desired result.

\section{Experiments}
\label{sec:experiments}

In this section we illustrate the performance of the proposed distributed dual averaging method through numerical experiments. The experiments are conducted on a cluster consisting of eight servers. Each server has two Intel Xeon 2.5GHz quad-core processors and 14GB of RAM. Our experiments make use of up to $n=64$ processes, with each process (i.e., node in the terminology of the previous sections) running on a separate core. The servers communicate over 1Gbps Ethernet connections.

We consider a multi-class classification task. In this task, each data point is a input-output pair, $(x,y) \in \mathcal{X} \times \mathcal{Y}$, with the input $x$ drawn from input space $\mathcal{X} \subseteq \mathbb{R}^d$, and the output $y$ drawn from a discrete output space $\mathcal{Y} = \{1, 2, \dots, M\}$ denoting the class of the data point. We use the MNIST digits dataset. Each input $x$ is a $28 \times 28$ greyscale image corresponding to one of the digits $0, 1, \dots, 9$. Hence, the input space has dimension $d = 784$ and there are $M=10$ classes.

We use multinomial logistic regression to address the multi-class classification task. In particular, the primal parameter $w$ is $M(d+1)$ dimensional and is broken into $M$ blocks, $w = [w_1^T, w_2^T, \dots, w_M^T]^T$, with $w_c \in \mathbb{R}^{d+1}$ being the parameters for the $c$th class. The loss function $f(w, (x,y))$ is defined as
\begin{equation*}
f(w, (x,y)) = \frac{1}{Z(w, x)} \exp\left\{w_{y,d+1} + \sum_{j=1}^d w_{y,d} x_d\right\},
\end{equation*}
where the normalization constant is given by
\begin{equation*}
Z(w, x) = \sum_{k=1}^M \exp\left\{w_{k,d+1} + \sum_{j=1}^d w_{k,d} x_d\right\}.
\end{equation*}

We report two sets of experiments, one in which we keep the mini-batch size $b$ fixed and vary the number of nodes, and the other where we fix the ratio $b/n$ (i.e., the number of data points processed by each node per round) and vary the number of nodes. In both experiments, gossip-based distributed dual averaging is implemented using the Matlab Parallel Computing Toolbox and Distributed Computing Server. The communication topology $G$ is obtained by sampling a random graph from the Erd\H{o}s-R\'{e}nyi model with $p=0.5$; all nodes have, on average, the same number of neighbours and the graph is approximately a constant-degree expander.

\subsection{Fixed mini-batch size}

In this experiment we fix the mini-batch size $b$ to be equal to the total number of data points available (roughly $50$k) and vary the number of nodes. Theorem~\ref{thm:gossipRegretBound} states that the regret after $T = \frac{m}{b+\mu}$ rounds should scale like $\mathbb{E}[R_n(m)] = \mathcal{O}(\sqrt{m})$. At each round, each node processes $b/n$ data points, so we have that $m = T \cdot \frac{b}{n} \cdot n = Tb$. Hence, the expected regret per data point should decay as $\frac{1}{m} \mathbb{E}[R_n(m)] = \mathcal{O}(\frac{1}{\sqrt{Tb}})$; i.e., this error only depends on the number of data points observed and not on the size of the network. We have observed exactly this in our experiments. Varying $n = 2^k$ for $k=2,3,\dots,6$, and plotting $\frac{1}{m} \mathbb{E}[R_n(T)]$ versus $T$, the resulting curves all lie directly on top of each other; the figure is omitted since it provides no additional insight. 

The benefit of using more nodes in this setting is that the workload of computing the batch gradient over $b$ data points is distributed over multiple processors, and so each round should require less time. Figure~\ref{fig:batch} shows $\frac{1}{m} \mathbb{E}[R_n(m)]$ as a function of run time for varying number of processors. Increasing the number of processors gives faster convergence to a lower average risk per data point, as expected. However, the benefit of adding more processors seems to obey a law of diminishing returns. This is related to the fact that using more processors requires more communication, and communication also requires time~\cite{TsianosNIPS2012}. 

In this experiment, the network only performs one consensus iteration per round, rather than the number prescribed by Theorem~\ref{thm:gossipRegretBound}. In other experiments, not reported here, we have observed that using additional consensus iterations does not significantly change the results.

\begin{figure}
\centering
\includegraphics[trim= 80 25 120 50,width=3in]{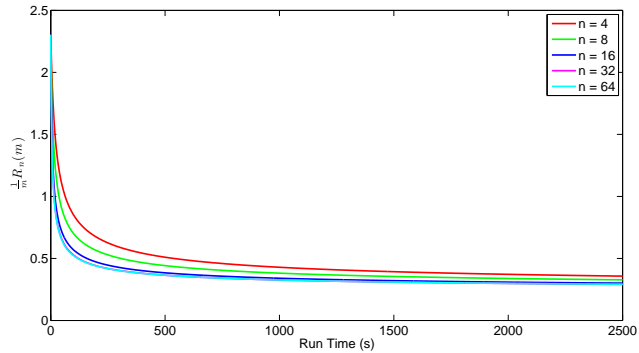}
\caption{Regret per data point as a function of time, for varying number of processors with fixed batch size $b=49984$.} \label{fig:batch}
\end{figure}

\subsection{Varying mini-batch size}

In the next set of experiments, we set the mini-batch size to be $b= 200 \times n$ (i.e., each node processes $200$ data points per round), and we again vary the number of nodes as $n=2^k$ for $k=2,3,\dots,6$. In this case, since each node processes the same amount of data per round, the time per round should be roughly constant, regardless of the network size. However, after a fixed number of rounds, $T$, the regret per data point $\frac{1}{m} \mathbb{E}[R_n(T)]$ should be lower for larger networks since they will have seen more cumulative data. In particular, with $b/n$ fixed constant, we have $\mathbb{E}[R_n(T)] = \mathcal{O}(\sqrt{T n})$, and so for $n' > n$, we expect that the ratio $R_n(T) / R_{n'}(T)$ should roughly behave as $\mathcal{O}(\sqrt{n / n'})$. The figure clearly shows that this behavior arises for sufficiently large $T$, after the transients have vanished and the $\mathcal{O}(\sqrt{m})$ behavior dominates. 

In these experiments, similar to above, only a single consensus iteration is performed over the network per round. In the proof of Theorem~\ref{thm:gossipRegretBound}, the number of consensus iterations given in \eqref{eqn:kthm} is required to ensure that the dual variables at all nodes are sufficiently close to the average of the dual variables, $\frac{1}{n} \sum_{i=1}^n z_i(t)$. In particular, this is the worst-case number of iterations required for any initial spread of the $z_i(t)$. However, in a typical execution of the algorithm, after processing one mini-batch of data the values at each node will not vary too greatly, and so the initial disagreement is not so large. Hence, we conjecture that an adaptive algorithm, which guarantees that the disagreement (i.e., deviation from the average) at each iteration is not too large, should be sufficient to obtain scaling results comparable to those in Theorem~\ref{thm:gossipRegretBound}.

\begin{figure}
\centering
\includegraphics[trim= 80 20 120 30,width=3in]{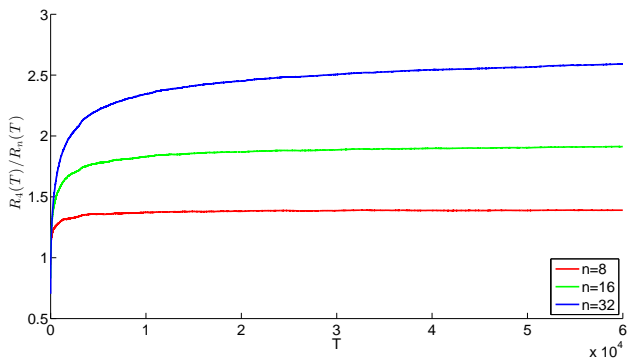}
\caption{Ratio of regrets $R_4(T) / R_n(T)$ versus number of rounds $T$ for $n=8, 16, 32$ with batch size $b = 200 n$ proportional to the network size. Note, $\sqrt{2} \approx 1.4$, $\sqrt{4} = 2$, and $\sqrt{8} \approx 2.8$.} \label{fig:online}
\end{figure}

\section{Discussion}
\label{sec:discussion}

In comparison to the distributed mini-batch algorithm~\cite{DekelMiniBatches}, we have shown that the proposed distributed dual averaging algorithm with approximate mini-batches achieves the same order-wise performance. In particular, if the communication topology is an expander graph and the mini-batch size is chosen to be $b = \Theta(\sqrt{m})$, then the expected regret is bounded as $\E[R_n(m)] = \order{ \sqrt{m} }$, and the time spent communicating in each round is of the order $\order{ \log n }$. This is comparable to the time spent performing exact distributed averaging using a protocol such as \textsc{AllReduce}, but gossip protocols can be considerably more robust to node and link failures~\cite{Jelasity2005,TsianosAllertonInvited}.

Raginsky et al.~\cite{Raginsky2011} and Yan et al.~\cite{Yan2013} study the related problem of distributed online convex optimization, where a different (deterministic) cost function $f_t$ is chosen by an adversary at each round, and the goal of the network is to minimize the regret with respect to the decision variable $w$ that would have minimized the overall cost $\sum_t f_t(w)$ in hindsight. The approach of Raginsky et al.~\cite{Raginsky2011} also builds on the dual averaging algorithm; however, rather than fusing information through gossip iterations, their approach involves the solution of a local dynamic program at each node.

This paper also makes contributions to the literature on gossip-based distributed optimization algorithms. The use of distributed averaging algorithms as a means of coordination in distributed optimization methods has a history tracing back to the seminal work of Tsitsiklis et al.~\cite{TsitsiklisAsyncGradOpt}. Distributed subgradient methods~\cite{nedicDistributedOptimization,distrStochSubgrOpt} that use gossip protocols to aggregate primal variables generally come with bounds guaranteeing that the optimality gap decays as $\order{ n^{3/2} / \sqrt{T}}$. The poor scaling in the network size appears since the analysis for these methods is effectively over the worst-case network topology. The distributed dual averaging algorithm described in~\cite{dualAveragingTAC} guarantees that the error is less than $\order{ \frac{\log(T \sqrt{n})}{\sqrt{T \cdot (1 - \lambda_2(P)) }} }$ after $T$ iterations. The number of transmissions per node is proportional to the number of iterations for the algorithms~\cite{nedicDistributedOptimization,distrStochSubgrOpt,dualAveragingTAC}. In contrast, for the proposed approach, the error decays as $\order{ \frac{1}{\sqrt{n T}} }$ where $T$ is the number of rounds, and in each round the network performs $\order{ \frac{\log(n)}{1 - \lambda_2(P) } }$ gossip iterations. Consequently, the overall communications required to reach accuracy $\epsilon$ is reduced by a factor of $\order{ n }$, as illustrated in Table~\ref{tab:comparison}. This is not surprising, since our analysis is based on stronger assumptions than~\cite{nedicDistributedOptimization,distrStochSubgrOpt,dualAveragingTAC}; in particular, we assume that the loss function has Lipschitz continuous gradients, whereas the previous work only studied loss functions which are convex, Lipschitz continuous, and have bounded gradients. The proposed method has a number of gossip iterations per round that depends on the spectral gap via $\frac{1}{1 - \lambda_2(P)}$. This can be immediately reduced to $\frac{1}{\sqrt{1 - \lambda_2(P)}}$ by using the accelerated distributed averaging method of~\cite{Oreshkin2010}, while only requiring that each node stores one extra copy of the dual variable being averaged.

\begin{table}[t]
\centering
\begin{tabular}{| c || c | c | c |} \hline
Method & Num.~Rounds ($T$) & Gossip Iter.~/ Round & Total Gossip Iter. \\ \hline \hline
\cite{nedicDistributedOptimization,distrStochSubgrOpt} & $\order{ \frac{n^3}{\epsilon^2} }$ & $1$ & $\order{ \frac{n^3}{\epsilon^2} }$ \\ \hline
\cite{dualAveragingTAC} & $\order{ \frac{1}{\epsilon^2} \frac{\log n}{\sqrt{1 - \lambda_2(P)}} }$ & $1$ & $\order{ \frac{1}{\epsilon^2} \frac{\log n}{\sqrt{1 - \lambda_2(P)}} }$ \\ \hline
This Paper & $\order{\frac{1}{n \epsilon^2}}$ & $\order{ \frac{\log n}{1 - \lambda_2(P)} }$ & $\order{ \frac{1}{\epsilon^2} \frac{\log n}{n (1 - \lambda_2(P))} }$ \\ \hline
\end{tabular}
\caption{Comparing the number of rounds and gossip iterations required to guarantee an optimality gap of $\epsilon$ for distributed stochastic optimization.} \label{tab:comparison}
\end{table}

As a side remark, we point out that faster gossip-based algorithms exist for more restricted function classes. Recall that in the setting of stochastic optimization the constant multiplying the $\order{ \frac{1}{\sqrt{T}} }$ term in single-processor optimality gap bounds is proportional to the standard deviation of the stochastic gradient magnitudes~\cite{Lan2012,xiaoDA}. In the deterministic, single-processor setting, and under the assumption that the loss function has Lipschitz continuous gradients, it is known that significantly faster rates are possible~\cite{NemirovskyYudinBook}. A gossip-based algorithm which achieves rates of $G_n(T) = \order{ \frac{1}{T^2} }$ is described in~\cite{FastDistributedGradMethods}. Similar to the gossip-based approach proposed in this paper, the approach in~\cite{FastDistributedGradMethods} also requires performing multiple gossip iterations per round. Unlike the approach in this paper, however, the number of gossip iterations per round increases as more rounds are performed; specifically, the number of gossip iterations required in round $t$ is $\order{ \frac{\log(t)}{1 - \lambda_2(P)} }$. In addition, the method of~\cite{FastDistributedGradMethods} must know the values of the gradient Lipschitz constant $K$ and the spectral gap $1 - \lambda_2(P)$ in order to determine a step-size sequence that achieves these bounds.

Finally, we note that similar rates for distributed stochastic optimization are achieved by a method based on a master-worker architecture~\cite{delayedDistrOpt}. In this architecture, the master maintains and updates the authoritative copy of the optimization variables. The workers compute mini-batch gradients and report these back to the master, where they are used to update the optimization variables. The challenge in this setup is that the workers report gradients evaluated at out-dated values of $w(t)$. The analysis in~\cite{delayedDistrOpt} involves bounding and controlling the additional error due to these delayed evaluations, and the resulting bounds are, at best, of the form $\Delta_n(T) = \order{\frac{n}{T} + \frac{1}{\sqrt{nT}}}$ or $\order{ \frac{1}{T^{2/3}} + \frac{1}{\sqrt{nT}} }$, depending on the version of the algorithm considered. There are also potential practical issues which must be addressed when implementing a master-worker architecture; for example, the master may become a bottleneck and single point of failure.

\section{Conclusion and Future Work}
\label{sec:conc}

The main message of this paper is that exact distributed averaging is not essential to achieve optimal rates with mini-batch methods for distributed online prediction and stochastic optimization. This opens the door to the use of gossip protocols for distributed averaging, which can be considerably more robust than exact distributed averaging protocols and can also be implemented in a completely asynchronous manner. A natural extension, and an important part of our ongoing work, is to analyze a completely asynchronous version of the distributed dual averaging algorithm with approximate mini-batch calculations. In the asynchronous setting, nodes may send and receive a random number of gossip messages, and they may also update their local copies of the decision variables at different rates. There are two major challenges in this setting. The first is to ensure that performance is not seriously degraded if, occasionally, an insufficient amount of gossiping occurs between subsequent updates. The second is to properly adapt the algorithm parameters $\beta(t)$ in such a manner that the asynchronous algorithm still converges at the optimal rate.

The algorithm and analysis in this paper are appropriate when the loss function is convex and has Lipschitz continuous gradients. For that class of problems, the proposed algorithm achieves the best rates possible. However, for other classes of loss functions --- in particular, for strongly convex loss functions --- it is known that better rates (e.g., $\E[R_1(m)] = \order{\log m}$) are achievable in the single-processor setting~\cite{largeScaleSGD,logRegRepeatedGames,TsengAccelerated}. At present, there are no known distributed optimization methods that achieve comparable rates while also exhibiting scalability in terms of the network size.

\appendix

\section{Proof of Lemma~\ref{lem:gossipRates}} \label{app:gossipRatesProof}

The proof of Lemma~\ref{lem:gossipRates} makes use of the following fact for stochastic matrices; see, e.g.,~\cite{HornJohnson}.

\begin{lem} \label{lem:eigenvalueBound}
Let $P$ be an $n \times n$ doubly-stochastic matrix with eigenvalues $\lambda_1(P) = 1 \ge \lambda_2(P) \ge \dots \ge \lambda_n(P) \ge 0$, and let $z$ be a vector in the $n$-dimensional unit simplex. Then 
\[
\norm{ z^T P^k - \frac{1}{n} \mathbf{1}_n^T }_1 \le \lambda_2(P)^k \sqrt{n} \;.
\]
\end{lem}

Recall that $\ybar = \frac{1}{n} \sum_{j=1}^n y_j^{(0)}$. After $k$ synchronous iterations of gossip, for all $i \in [n]$ it holds that 
\begin{align*}
y_i^{(k)} - \ybar &= \sum_{j=1}^n [P^k]_{i,j} (y_j^{(0)} - \ybar) \\
&= \sum_{j=1}^n \left([P^k]_{i,j} - \frac{1}{n}\right) y_j^{(0)} \\
&= \sum_{j=1}^n \left([P^k]_{i,j} - \frac{1}{n}\right) \left(y_j^{(0)}  - \ybar \right),
\end{align*}
where the last line follows since $P^k$ is also doubly stochastic, and so $\sum_{j=1}^n [P^k]_{i,j} = \sum_{i=1}^n \frac{1}{n} = 1$.

Applying the triangle inequality gives
\begin{align*}
\norm{ y_i^{(k)} - \ybar } &\le \sum_{j=1}^n \abs{ [P^k]_{i,j} - \frac{1}{n} } \cdot \norm{ y_j^{(0)} - \ybar } \\
&\le \left(\max_j \norm{ y_j^{(0)} - \ybar } \right) \norm{ \bm{e}_i^T P^k - \frac{1}{n} \mathbf{1}_n }_1,
\end{align*}
where $\bm{e}_i$ denotes the $i$th canonical vector, and where $\norm{\cdot}_1$ denotes the $\ell_1$ norm; $\norm{x}_1 = \sum_i |x_i|$.\footnote{I.e., $\bm{e}_i \in \{0,1\}^n$ has $j$th entry $[\bm{e}_i]_j = 1$ if $j = i$, and $[\bm{e}_i]_j = 0$ otherwise.} Clearly $\bm{e}_i$ lies in the unit simplex. Applying the result from Lemma~\ref{lem:eigenvalueBound}, and noting that $\lambda_2(P^k) = \lambda_2(P)^k$, we have that
\begin{align*}
\norm{ y_i^{(k)} - \ybar } &\le 2 \sqrt{n} \lambda_2(P)^k \cdot \max_j \norm{ y_j^{(0)} - \ybar }.
\end{align*}
It follows that $\norm{ y_i^{(k)} - \ybar } \le \delta$ if
\[
k \ge \frac{\log\left( \frac{1}{\delta} \cdot 2 \sqrt{n} \cdot \max_j \norm{ y_j^{(0)} - \overline{y} }\right)}{\log(1 / \lambda_2(P))} \;.
\]
The claim of the lemma follows by recalling that $\log(1/x) \ge 1 - x$ for $x > 0$.

\section{Proof of Lemma~\ref{lem:errorProperties}}
\label{sec:errorPropertiesProof}

We  prove the four relationships in the order they are stated in the lemma. First consider $\Exp{}{\dprod{q(t)}{w^* - \overline{w}(t)}}$. Notice that $\overline{w}(t)$ is a deterministic function given $X(t-1)$ so
\begin{align*}
\Exp{}{f(\overline{w}(t), x_i(t))} &= \Exp{}{ \Exp{}{ f(\overline{w}(t), x_i(t)) \vert X(t-1) } } = \Exp{}{F(\overline{w}(t))}.
\end{align*}
Also, under Assumption~\ref{ass:lossFunction}, the gradient and expectation operators commute~\cite{RockafellarWetts}, and so we have
\begin{equation*}
\Exp{}{\dprod{\nabla f(\overline{w}(t), x_j(t))}{ \overline{w}(t) }} = \Exp{}{\dprod{\nabla F(\overline{w}(t))}{ \overline{w}(t) }}.
\end{equation*}
Therefore, for the first expression in the lemma, we get
\begin{align*}
\Exp{}{\dprod{q(t)}{w^* - \overline{w}(t)}} &= \Exp{}{\dprod{\widehat{g}(t) - \nabla F(\overline{w}(t))}{w^* - \overline{w}(t)} } \\
&= \Exp{}{\dprod{\frac{1}{b}\sum_{i=1}^n \sum_{s=1}^{\frac{b}{n}} \nabla f(\overline{w}(t), x_i(t,s)) - \nabla F(\overline{w}(t))}{w^* - \overline{w}(t)} } \\
&= \Exp{}{\dprod{\frac{1}{b}\sum_{i=1}^n \sum_{s=1}^{\frac{b}{n}} \nabla F(\overline{w}(t)) - \nabla F(\overline{w}(t))}{w^* - \overline{w}(t)} } \\
&= 0.
\end{align*}

For the second expression, using similar conditioning arguments gives
\begin{align*}
\Exp{}{\dprod{r(t)}{w^* - \overline{w}(t)}} &= \Exp{}{ \dprod{\frac{1}{b}\sum_{i=1}^n \sum_{s=1}^{\frac{b}{n}} \nabla f(w_i(t), x_i(t,s)) - \frac{1}{b}\sum_{i=1}^n \sum_{s=1}^{\frac{b}{n}} \nabla f(\overline{w}(t), x_i(t,s))}{w^* - \overline{w}(t)} } \\
&= \Exp{}{ \dprod{\frac{1}{b}\sum_{i=1}^n \sum_{s=1}^{\frac{b}{n}} \nabla F(w_i(t)) - \frac{1}{b}\sum_{i=1}^n \sum_{s=1}^{\frac{b}{n}} \nabla F(\overline{w}(t))}{w^* - \overline{w}(t)} } \\
&= \frac{1}{n}\sum_{i=1}^n \Exp{}{ \dprod{ \nabla F(w_i(t)) -   \nabla F(\overline{w}(t))}{w^* - \overline{w}(t)} } \;.
\end{align*}
Then, applying the Cauchy-Schwarz inequality and using Assumptions~\ref{ass:constraintSet} and~\ref{ass:lossFunction} gives
\begin{align*}
\Exp{}{\dprod{r(t)}{w^* - \overline{w}(t)}} & \leq \frac{1}{n}\sum_{i=1}^n \Exp{}{ \norm{ \nabla F(w_i(t)) -   \nabla F(\overline{w}(t))} \cdot \norm{w^* - \overline{w}(t)} } \\
& \leq \frac{1}{n}\sum_{i=1}^n \Exp{}{ K \norm{ w_i(t) - \overline{w}(t)} \cdot D } \;.
\end{align*}
Finally, applying Lemma~\ref{lem:proxProjectionLipschitz} gives the desired result,
\begin{align*}
\Exp{}{\dprod{r(t)}{w^* - \overline{w}(t)}} & \leq \frac{K D}{n \beta(t)}\sum_{i=1}^n \Exp{}{  \norm{ z_i(t) -  \overline{z}(t)}} \;.
\end{align*}

The third expression follows by again using similar conditioning arguments and by applying Assumption~\ref{ass:boundedGradientVariance}. In particular, recalling that the data are i.i.d., and hence $x_i(t,s)$ is independent of $x_j(t,s)$, we have
\begin{align*}
\Exp{}{\lVert q(t) \rVert^2} &= \Exp{}{\norm{\hat{g}(t) - \nabla F(\overline{w}(t)) }^2} \\
&= \Exp{}{\norm{\frac{1}{b} \sum_{i=1}^n \sum_{s=1}^{\frac{b}{n}} \nabla f(\overline{w}(t), x_i(t,s)) - \nabla F(\overline{w}(t)) }^2} \\
&= \Exp{}{\norm{\frac{1}{n} \sum_{i=1}^n \frac{n}{b} \sum_{s=1}^{\frac{b}{n}} \nabla f(\overline{w}(t), x_i(t,s)) - \nabla F(\overline{w}(t)) }^2} \\
&= \Exp{}{\frac{1}{n^2} \sum_{i=1}^n  \sum_{j=1}^n  \dprod{\frac{n}{b} \sum_{s=1}^{\frac{b}{n}} \nabla f(\overline{w}(t), x_i(t,s)) - \nabla F(\overline{w}(t)) }{\frac{n}{b} \sum_{s=1}^{\frac{b}{n}} \nabla f(\overline{w}(t), x_j(t,s)) - \nabla F(\overline{w}(t))}} \\
&= \frac{1}{n^2} \sum_{i=1}^n \sum_{j=1}^n \Exp{}{ \dprod{\frac{n}{b} \sum_{s=1}^{\frac{b}{n}} \nabla f(\overline{w}(t), x_i(t,s)) - \nabla F(\overline{w}(t)) }{\frac{n}{b} \sum_{s=1}^{\frac{b}{n}} \nabla f(\overline{w}(t), x_j(t,s)) - \nabla F(\overline{w}(t))}} \\
&= \frac{1}{n^2} \sum_{i=1}^n  \Exp{}{ \norm{\frac{n}{b} \sum_{s=1}^{\frac{b}{n}} \nabla f(\overline{w}(t), x_i(t,s)) - \nabla F(\overline{w}(t)) }^2} \;.
\end{align*}
Then, applying the triangle inequality and using Assumption~\ref{ass:boundedGradientVariance} gives
\begin{align*}
\Exp{}{\lVert q(t) \rVert^2}&\leq \frac{1}{n^2} \sum_{i=1}^n  \frac{n^2}{b^2} \sum_{s=1}^{\frac{b}{n}} \Exp{}{ \norm{\nabla f(\overline{w}(t), x_i(t,s)) - \nabla F(\overline{w}(t)) }^2} \\
&\leq \frac{1}{b^2} \sum_{i=1}^n  \sum_{s=1}^{\frac{b}{n}} \sigma^2 \\
& \leq \frac{\sigma^2}{b} \;.
\end{align*}

Finally, to verify the fourth expression, we first multiply and divide by $n$ and use the triangle inequality to see that
\begin{align*}
\Exp{}{\lVert r(t) \rVert^2} &= \Exp{}{\norm{ \frac{1}{b}\sum_{i=1}^n \sum_{s=1}^{\frac{b}{n}}\nabla f(w_i(t), x_i(t,s)) - \nabla f(\overline{w}(t), x_i(t,s))}^2} \\
&= \Exp{}{\norm{ \frac{1}{n}\sum_{i=1}^n \frac{n}{b}\sum_{s=1}^{\frac{b}{n}}\nabla f(w_i(t), x_i(t,s)) - \nabla f(\overline{w}(t), x_i(t,s))}^2} \\
&\leq \Exp{}{  \left(  \frac{1}{n}\sum_{i=1}^n \frac{n}{b}\sum_{s=1}^{\frac{b}{n}} \norm{\nabla f(w_i(t), x_i(t,s)) - \nabla f(\overline{w}(t), x_i(t,s))} \right)^2} \;.
\end{align*}
Then applying Assumption~\ref{ass:lossFunction} and Lemma~\ref{lem:proxProjectionLipschitz} and simplifying gives the desired result,
\begin{align*}
\Exp{}{\lVert r(t) \rVert^2} &\leq \Exp{}{  \left(  \frac{1}{n}\sum_{i=1}^n \frac{n}{b}\sum_{s=1}^{\frac{b}{n}} K \norm{ w_i(t) - \overline{w}(t)} \right)^2} \\
&\leq \Exp{}{  \left(  \frac{1}{n}\sum_{i=1}^n  \frac{K}{\beta(t)} \norm{ z_i(t) - \overline{z}(t)} \right)^2} \\
&= \Exp{}{  \left(  \frac{K}{n \beta(t)}\sum_{i=1}^n   \norm{ z_i(t) - \overline{z}(t)} \right)^2} \\
&= \Exp{}{   \frac{K^2}{n^2 \beta(t)^2} \sum_{i=1}^n \sum_{j=1}^n   \norm{ z_i(t) - \overline{z}(t)} \norm{ z_j(t) - \overline{z}(t)} }  \;.
\end{align*}

\bibliographystyle{abbrv}
\bibliography{../PhDThesis/References}

\end{document}